\documentclass[pra,aps,twocolumn,superscriptaddress,longbibliography,nofootinbib]{revtex4-2}

\usepackage{amsmath,amsthm,amssymb,amsfonts,setspace}
\usepackage{bm}
\usepackage{bbm}
\usepackage{float}
\usepackage{epsfig,color}
\usepackage{tikz}
\usepackage{bbold}
\usepackage{xfrac}
\usetikzlibrary{positioning,arrows.meta,decorations.pathreplacing}

\usepackage[dvipsnames]{xcolor}
\usepackage[colorlinks]{hyperref}
\hypersetup{
    allcolors  = {teal!75!blue},
}
\usepackage{mathtools}
\usepackage{cleveref} 
\usepackage{fontenc}
\usepackage{soul}

\usepackage[normalem]{ulem} 
\usepackage{graphicx}

\crefformat{equation}{Eq.~(#2#1#3)} 
\crefformat{section}{Section~#2#1#3}
\crefformat{theorem}{Theorem~#2#1#3}
\crefformat{lemma}{Lemma~#2#1#3}
\crefformat{appendix}{Appendix~#2#1#3}
\Crefformat{equation}{Equation~(#2#1#3)}
\crefformat{figure}{Fig.~#2#1#3}
\crefformat{corollary}{Corollary~#2#1#3}
\crefformat{proposition}{Proposition~#2#1#3}
\Crefformat{proposition}{Proposition~#2#1#3}
\crefformat{exa}{Example~#2#1#3}
\Crefformat{exa}{Example~#2#1#3}
\crefformat{definition}{Definition~#2#1#3}
\Crefformat{definition}{Definition~#2#1#3}
\crefrangeformat{equation}{Eqs.~#3(#1)#4--#5(#2)#6}

\usepackage{etoolbox}

\theoremstyle{plain}
\newtheorem{theorem}{Theorem}[section]
\newtheorem{corollary}[theorem]{Corollary}

\theoremstyle{definition}
\newtheorem{definition}[theorem]{Definition}
\AtEndEnvironment{definition}{\hfill$\diamond$}
\newtheorem{exa}[theorem]{Example}

\newcommand{\bra}[1]{\langle{#1}\vert} 
\newcommand{\ket}[1]{\vert{#1}\rangle} 

 \newcommand{\ketbra}[2]{\left\vert{#1}\right\rangle\!\left\langle{#2}\right\vert}
 \newcommand{\braket}[2]{\left\langle{#1}\middle\vert  {#2} \right\rangle}

\renewcommand{\L}{\left(} 
\newcommand{\R}{\right)} 

\renewcommand{\r}{\mathbf{r}} 

\begin{document}

\title{Coherence as a resource for $N$-box and quantum pigeonhole paradoxes}

\newcommand{\inl}{INL -- International Iberian Nanotechnology Laboratory, Av. Mestre Jos\'{e} Veiga s/n, 4715-330 Braga, Portugal}
\newcommand{\inlshort}{INL -- International Iberian Nanotechnology Laboratory, Braga, Portugal}
\newcommand{\uff}{Instituto de F\'{i}sica, Universidade Federal Fluminense, Av. Gal. Milton Tavares de Souza s/n, Niter\'{o}i -- RJ, 24210-340, Brazil}
\newcommand{\uffshort}{Instituto de F\'{i}sica, Universidade Federal Fluminense, Niter\'{o}i -- RJ, Brazil}
\newcommand{\cfum}{Centro de F\'{i}sica, Universidade do Minho, Campus de Gualtar, 4710-057 Braga, Portugal}
\newcommand{\cfumshort}{Centro de F\'{i}sica, Universidade do Minho, Braga, Portugal}
\newcommand{\ulm}{Institute of Theoretical Physics, Ulm University, Albert-Einstein-Allee 11, 89081 Ulm, Germany} 
\newcommand{\iqst}{Center for Integrated Quantum Science and Technology (IQST), 89081 Ulm, Germany} 

\author{Som Kanjilal}
\email{somkanjilal@gmail.com}
\affiliation{\inlshort}

\author{Rafael Wagner}
\email{rafael.wagner@uni-ulm.de}
\affiliation{\ulm}
\affiliation{\iqst}

\author{Ernesto F. Galv\~ao}
\email{ernestogalvao@id.uff.br}
\affiliation{\inlshort}
\affiliation{\uffshort}

\date{\today}

\begin{abstract}
    Pre- and post-selection (PPS) paradoxes are striking demonstrations of quantum nonclassicality. \emph{Logical} PPS paradoxes, where inferences made with the Aharonov--Bergmann--Lebowitz (ABL) rule are exactly 0 or 1, are linked to contextuality. Non-logical paradoxes lack this strong signature. In this work, we analyse more general, non-logical PPS scenarios involving mixed pre- and post-selected states. We show that two such scenarios, the $N$-box and quantum pigeonhole paradoxes, require coherence of both pre- and post-selected states in the basis of the intermediate measurement. This is done by showing each paradox holds if and only if there is weak-value anomaly for a single, paradox-specific operator, together with the fact that weak-value anomaly requires coherence. This clarifies the role of different notions of nonclassicality in these scenarios, highlighting the required quantitative departures from (strictly) classical explanations provided by incoherent sub-theories of quantum theory.
\end{abstract}

\maketitle

\section{Introduction}\label{sec: intro}

Pre- and post-selection (PPS) paradoxes~\cite{aharonov1990properties,MARONEY2,aharonov2002time,LeiferSpekkensontology,kastner2003nature,aharonov2005quantum} are counterintuitive predictions of quantum theory. These are physical scenarios where a system is prepared, subjected to either strong or weak (often counter-factually reasoned) intermediate measurements, and then post-selected on a specific final outcome. The seemingly paradoxical features arise when modeling such scenarios using quantum theory \emph{but} attempting to provide a classical account on the observed statistics that assumes observables of the system to have simultaneously well-defined, pre-existing properties. One of the most notable examples is provided by the paradigmatic three-box paradox~\cite{albert1985curious,bub1986curious,albert1986comment}, where a \emph{single} particle that is suitably prepared and later post-selected appears to be found with certainty in each of \emph{two} different boxes, when checked one box at a time. 

There is a rich landscape of PPS paradoxes beyond the three-box example just mentioned. To start, one can generalize the scenario of the three-box paradox to an arbitrary number $N+1$ of boxes, where a single particle can be found (when checked between the pre- and post-selection stages) in any one of the $N$ tested boxes~\cite{aharonov1991complete,leavens2006general}. Other examples are the quantum Cheshire cat~\cite{aharonov2013cheshire,aharonov2021dynamical,hance2023contextuality,hance2024dynamicalcheshire,denkmayr2014observation,duprey2018quantum}, Hardy's paradox~\cite{hardy1992quantum,abramsky2012logical,aharonov2016weakvalues,aharonov2002revisiting,santos2021conditions}, and the quantum pigeonhole principle~\cite{aharonov2016quantum}. In each of these cases, the paradox forces us to revisit classical intuitions about the properties of physical systems: The quantum Cheshire cat is supposed to reveal a counterintuitive separation of properties from their carriers; Hardy's paradox demonstrates the failure of Bell's notion of local causality without any inequality violation; and the quantum pigeonhole principle shows how pre- and post-selection can invalidate classical counting reasoning.

The class of \emph{logical} PPS paradoxes~\cite{LeiferspekkensPPScontextuality}, which account for the case when the probabilities for the intermediate measurement conditioned on \emph{both} the pre- and post-selections via the
Aharonov--Bergmann--Lebowitz (ABL)~\cite{aharonov1964time} rule are either 0 or 1, have a sharp nonclassical status: logical PPS paradoxes are proofs of quantum contextuality~\cite{LeiferspekkensPPScontextuality,PuseyLeifer2015,LeiferSpekkensontology,tollaksen2007pre,hance2023contextuality,waegell2017contextuality,waegell2017confined}, a stringent criterion for nonclassicality in quantum theory~\cite{budroni2022kochenspecker}. Put simply, contextuality can be viewed as the impossibility of describing a given experiment's outcome statistics using a classical probabilistic model where the result of any measurement is independent of the broader set of compatible measurements being concurrently 
performed. 

If we set contextuality aside and instead examine the weaker nonclassicality condition that the quantum states relevant for a logical PPS scenario be \emph{coherent}~\cite{streltsov2017colloquium,baumgratz2014quantifying}, this aspect is also relatively well understood. For instance, Pusey and Leifer~\cite{PuseyLeifer2015} demonstrated that logical PPS paradoxes require the pre-selected state $\vert \psi_I \rangle$ and the post-selected state $\vert \psi_F \rangle$ to be non-orthogonal. This is a well-known signature of basis-independent coherence of the pair of states~\cite{designolle2021set,galvao2020quantum}.  Furthermore, Aharonov and Cohen~\cite{aharonov2016weakvalues}  linked these paradoxes to nonclassical \emph{weak values}~\cite{aharonov1988result,tamir2013introduction,dressel2014colloquium}, which are   connected to both coherence~\cite{WG23,wagner2024quantumcircuits,hofmann2009resolutionquantumparadoxesweak,hofmann2011ontherole,hofman2015quantum} and contextuality~\cite{pusey2014anomalous,kunjwal2019anomalous}. More recently, focusing on the quantum Cheshire cat, Hance \emph{et al.}~\cite{hance2023contextuality} arrived at a similar conclusion, demonstrating that the negativity of a set of constructed weak values signals the necessity of quantum coherence to the paradox.

In contrast to the logical case, the operational formalization and nonclassicality analysis of \emph{non-logical} PPS paradoxes have remained largely unexplored. To the best of our knowledge, the only known instance is the ``gamified'' three-box paradox considered by  Maroney~\cite{MARONEY2}. Maroney demonstrated that the nonclassicality of this paradox is tied to scenarios satisfying an operational non-disturbance condition~\cite{Maroney1,MARONEY2}. This condition stipulates that the final post-selection cannot provide information about the outputs of an intermediate measurement that was performed, and it is formally connected to Leggett and Garg's notions of macrorealism via the measurement non-invasiveness condition~\cite{leggett1985quantum,emary2013leggett,vitagliano2023leggett,schmid2024reviewreformulation,hermens2018constraints}. Through this connection, Maroney formally related the signatures of this non-logical three-box paradox to the violation of a Leggett--Garg inequality~\cite{emary2013leggett}.
 
In this work, we extend Maroney's non-logical three-box PPS scenario to allow for an arbitrary number of boxes and mixed pre- and post-selected states, and we show that the same generalization applies to the quantum pigeonhole PPS scenario. We show that, in each generalized scenario, both the pre- and post-selected states must be coherent with respect to the intermediate measurement basis, and we identify the anomaly of a single weak value as the operational, quantifiable witness of this coherence. A paradox arises precisely when this weak value is anomalous. Equivalently, the paradox arises when the Bargmann invariants~\cite{bargmann1964note}---multivariate traces of states, invariant under changes of basis, whose values witness quantum coherence~\cite{oszmaniec2024measuring,fernandes2024unitary,wagner2024quantumcircuits}---have negative real part. 

We further discuss how these results are consistent with  Maroney's identification of the three-box paradox with a Leggett--Garg violation~\cite{MARONEY2}. We reinterpret this result through Schmid's~\cite{schmid2024reviewreformulation} recent reformulation of Leggett and Garg's notion of  macrorealism as a theory-independent strict form of classicality. Schmid's proposal is related to the form of coherence considered here when specialized to quantum theory. 

Our results solidify our understanding of the role of coherence in PPS paradoxes, and provide an instructive path toward understanding the role of other stringent notions of classicality such as generalized contextuality~\cite{spekkens2005contextuality} in such paradoxes beyond the logical case.

\emph{Outline.---}The paper is organized as follows. In \cref{sec:scenario} we review the operational $N$-box PPS scenario, formalize it in Def.~\ref{defNbox}, and isolate the two operational constraints needed for our results (\cref{fig:scenario}). \Cref{sec:coherence} introduces (set) coherence and anomalous weak values, and situates our results relative to the known logical-case connection between non-orthogonality,  coherence, and contextuality~\cite{PuseyLeifer2015,aharonov2016weakvalues}. In \cref{sec:results} we prove that a coherence-free subtheory cannot yield an $N$-box paradox (\cref{thm:nonlogical}), and that nonclassicality is governed by negativity in the weak value of the excluded box (\cref{coro:anomalous}). \Cref{sec:LG} revisits Maroney's~\cite{MARONEY2} identification of the three-box paradox with a Leggett--Garg violation, linking our results to Schmid's~\cite{schmid2024reviewreformulation} recent reformulation of macrorealism as a theory-independent strict form of classicality. In Sec.~\ref{sec:pigeon} we carry out the same programme for the PPS scenario associated to the quantum pigeonhole principle: we formalize the scenario and its non-logical extension in Def.~\ref{defPigeon} and isolate its own minimal operational conditions that are expected to be satisfied for a paradoxical reasoning following Maroney's~\cite{MARONEY2} analysis, before showing, in the following section, that a coherence-free subtheory likewise cannot yield a pigeonhole paradox (\cref{thm:pigeon-nonlogical}), and that nonclassicality is governed by the anomalous weak value of a single operator (\cref{coro:pigeon-anomalous}). \Cref{sec:conclusions} concludes with a discussion of our results and open directions.

\emph{Notation.---}For an arbitrary quantum system described by a Hilbert space $\mathcal{H}$, the state-space of possible quantum states is denoted $\mathcal{D}(\mathcal{H})$. All Hilbert spaces are assumed to be finite, and chosen as $\mathcal{H} = \mathbbm{C}^d$ for some $d$. A quantum observable $A$ is a Hermitian operator. Positive operator-valued measures are denoted $\mathcal{M}=\{E_i\}_i$ and if the elements are orthogonal projectors we use $\mathcal{M} = \{\Pi_i\}_i$ instead. Mixed states are denoted as $\rho$ and pure states are denoted as $\psi = \vert \psi \rangle \langle \psi \vert$. The weak value of an observable $A$ relative to pure pre- and post-selected states $\vert \psi_I \rangle$ and $\vert \psi_F \rangle$ is denoted as $\langle A \rangle_w \equiv \langle \psi_F \vert A \vert \psi_I \rangle/\langle \psi_F\vert \psi_I \rangle$.

\section{The $N$-box PPS paradox}\label{sec:scenario}

We start by discussing the steps of the PPS  
paradox considered in Refs.~\cite{aharonov1990properties,LeiferspekkensPPScontextuality,PuseyLeifer2015,Maroney1,MARONEY2} in its generalized version with $N$ boxes as in Refs.~\cite{aharonov1991complete,leavens2006general}. 

The description of the scenario starts by assuming we prepare an $(N+1)$-level system $\mathcal{H} = \mathbbm{C}^{N+1}$ representing a discretization of possible locations a particle can be in. For that system, we let $B = \{\vert b\rangle\}_{b=1}^{N+1}$ be an arbitrarily chosen orthonormal basis, whose elements are called \emph{boxes}. We thus model whether the particle is in box $b$ by writing that its quantum state is $\psi_b = \vert b \rangle \langle b \vert $. Naturally, measurements of the form $\mathcal{M}_b=\{\Pi_b,\mathbb{1}-\Pi_b\}$ are interpreted---following classical intuition that we will shortly see leads to  paradoxical interpretations---as checking whether the particle is in box $b$. This is so because if the particle is in box $b$, represented by the state $\psi_b = \vert b \rangle \langle b \vert$, a measurement $\mathcal{M}_b$ will return $\Pi_b$ with certainty. If instead the particle is in box $b' \neq b$,  $\psi_{b'} = \vert b'\rangle \langle b' \vert$, we find that a measurement $\mathcal{M}_b$ will return $\mathbb{1}-\Pi_b \equiv \Pi_b^{\perp}$ with certainty.

Operationally, the scenario is summarized in Fig.~\ref{fig:scenario}. We start by preparing the state of our system in an arbitrary quantum state $\rho_I \in \mathcal{D}(\mathbbm{C}^{N+1})$. At an intermediate time, an experimenter performs one of $N$ possible dichotomic sharp non-destructive measurements $\{\mathcal{M}_b\}_{b=1}^N$, each asking the question ``\emph{is the particle in box $b$?}''~\cite{Maroney1}. We exclude the possibility of performing the intermediate measurement $\mathcal{M}_{N+1}$. The above procedure is followed by a post-selection on another arbitrary quantum state $\rho_F \in \mathcal{D}(\mathbbm{C}^{N+1})$.

\begin{figure}[t]
\centering
\begin{tikzpicture}[>=Latex, font=\small,
  stage/.style={rounded corners, draw, thick, align=center,
                minimum height=1.35cm, text width=2.1cm, fill=blue!5}]
\node[stage] (prep) {Pre-selection\\$\rho_I$};
\node[stage, right=0.6cm of prep] (mid) {Intermediate\\measurement\\$\mathcal{M}_b$};
\node[stage, right=0.6cm of mid] (post) {Post-selection\\$\rho_F$};
\draw[->,thick] (prep)--(mid);
\draw[->,thick] (mid)--(post);
\draw[decorate, decoration={brace, amplitude=6pt, raise=4pt}, thick]
   (prep.north west) -- (post.north east)
   node[midway, above=10pt]{Operational non-disturbance~\eqref{eq:op_non_dist}};
\node[below=0.45cm of mid, align=center, text width=8cm, font=\footnotesize]
   {$\mathcal{M}_b=\{\Pi_b,\mathbb{1}-\Pi_b\}$, \ $b=1,\dots,N$ \\ (Box $N+1$ is never measured)};
\end{tikzpicture}
\caption{\textbf{The $N$-box PPS scenario.} A system---interpreted as a particle which can be in $N+1$ distinct boxes---is prepared in quantum state $\rho_I$. An intermediate measurement $\mathcal{M}_b = \{\Pi_b,\mathbb{1}-\Pi_b\}$ with $\Pi_b = \vert b \rangle \langle b \vert$, interpreted as checking whether the particle is found in box $b$, follows. The measurement checking whether the particle is found in box $N+1$ is excluded. Finally, the system is post-selected on $\rho_F$. In all runs, we assume that the post-selection onto $\rho_F$ (successful or not) cannot yield information about which intermediate measurement was performed, a constraint captured by the non-disturbance condition from Eq.~\eqref{eq:op_non_dist}.
\label{fig:scenario}
}
\end{figure}
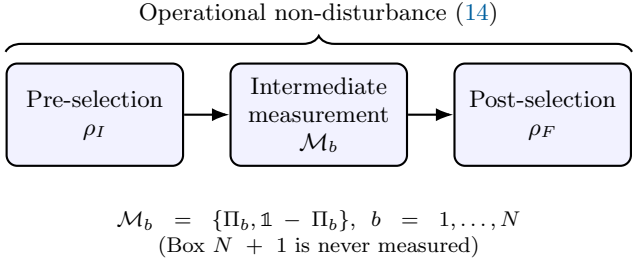

Let us denote the post-selection measurement $\mathcal{M}_F=\{\rho_F,\mathbb{1}-\rho_F\}$ and write ``${\rm pass}$'' to represent that the post-selection was successful (i.e. the output of $\mathcal{M}_F$ was $\rho_F$) and ``${\rm fail}$'' otherwise. Similarly, we write $b$ to label that the intermediate measurement $\mathcal{M}_b$ resulted in $\Pi_b =  \vert b \rangle \langle b \vert $ and $\neg b$ otherwise. Whether post-selection is successful is, itself, going to depend on the probabilistic predictions of quantum theory that follow from Lüders' state-update rule. The joint probabilities associated to both the intermediate and post-selection measurements are given as:
\begin{align}
\label{eq:Nbox11}
p(b,{\rm pass} \mid \rho_I,\mathcal{M}_b,\mathcal{M}_F) & = \mathrm{Tr}\left[\rho_F\Pi_{b}\rho_I\Pi_{b}\right],\\
\label{eq:Nbox01}
p(\neg b,{\rm pass} \mid \rho_I,\mathcal{M}_b,\mathcal{M}_F) & = \mathrm{Tr}\left[\rho_F\Pi_{b}^\perp\rho_I\Pi_{b}^\perp\right],\\
\label{eq:Nbox10}
p(b,{\rm fail} \mid \rho_I, \mathcal{M}_b,\mathcal{M}_F) & = \mathrm{Tr}\left[(\mathbb{1}-\rho_F)\Pi_{b}\rho_I\Pi_{b}\right],\\
\label{eq:Nbox00}
p(\neg b,{\rm fail} \mid \rho_I, \mathcal{M}_b,\mathcal{M}_F) & = \mathrm{Tr}\left[(\mathbb{1}-\rho_F) \Pi_{b}^\perp\rho_I \Pi_{b}^\perp\right].
\end{align}  
Above, we have denoted $\Pi_{b}^\perp = \mathbb{1}-\Pi_b$. Whenever $\rho_I$, $ \mathcal{M}_b$, and $\mathcal{M}_F$ are clear from the context we simply write $p(i,f) \equiv p(i,f \mid \rho_I,\mathcal{M}_b,\mathcal{M}_F)$ with $i \in \{b,\neg b\}$ and $f \in \{{\rm pass}, {\rm fail}\}$. Moreover, the probability of successful post-selection when no intermediate measurement is performed can be written as 
\begin{equation}
\label{eq:null}
p({\rm pass} \mid \rho_I,\mathcal{M}_F)=\mathrm{Tr}\left[\rho_{F}\rho_{I}\right].
\end{equation}

Inferential reasoning in pre- and post-selection scenarios has intrinsic distinctions relative to typical operational prepare-transform-measure scenarios. One key difference is that we usually make inferences conditioning only on the initial state of our system. In PPS scenarios we instead make inferences conditioning on \emph{both} the initial and final state, provided that a post-selection was successful. This significantly affects counterfactual reasoning. We are interested in making inferences about what would have happened at the intermediate time, had we implemented a measurement $\mathcal{M}_b$, but this is always \emph{conditioned} on the fact that we post-select the state of our system onto $\rho_F$ at the end. 

The correct way to infer the probability that a given output $\Pi_b$ of an intermediate measurement $\mathcal{M}_b$ occurred relative \emph{both} to an initial state $\rho_I$ \emph{and} a successful post-selection onto $\rho_F$ was first described by Aharonov, Bergmann, and Lebowitz in Ref.~\cite{aharonov1964time}, and is given by 
\begin{align}\label{eq:ABL_rule}
p(b \vert \rho_I,\rho_F,\mathcal{M}_b)
=
\frac{p(b,{\rm pass})}
        {p(b,{\rm pass})+p(\neg b,{\rm pass})}.
\end{align}
This is known as the \emph{ABL rule}. Note that above $p(b,{\rm pass}) \equiv p(b,{\rm pass} \mid \rho_I, \mathcal{M}_b,\mathcal{M}_F)$ and similarly for $p(\neg b,{\rm pass})$.

\subsection{Three-box case}\label{subsec:three_box_case}

Reasoning with the ABL rule, together with the classical interpretational baggage we have attached to the intermediate measurements, leads us to paradoxical conclusions. The paradigmatic three-box example takes a three level quantum system pre-selected to be in the state $\rho_I=\ket{\psi_I}\bra{\psi_I}$ where \begin{equation}\label{eq:pre_selected_threeboxclassical}
\ket{\psi_I}=\frac{1}{\sqrt{3}}(\ket{1}+\ket{2}+\ket{3})
\end{equation}
 and post-selected in the state  $\rho_F=\ket{\psi_F}\bra{\psi_F}$ where 
\begin{equation}\label{eq:post_selected_threeboxclassical}
 \ket{\psi_F}=\frac{1}{\sqrt{3}}(\ket{1}+\ket{2}-\ket{3}).
\end{equation}
In this case, we have a non-zero rate of successful post-selection when no measurement occurs, given by 
\begin{equation}\label{eq:3null}
    p({\rm pass} \mid \rho_I,\mathcal{M}_F) = \mathrm{Tr}[\rho_I \rho_F] = \vert \langle \psi_I \vert \psi_F \rangle  \vert^2 =\sfrac{1}{9}.
\end{equation} 
The joint probability that a post-selection was successful \emph{and} that upon implementing the intermediate measurement $\mathcal{M}_b$ with $b \in \{1,2\}$ is 
\begin{align}
\label{eq:3box11}
p(b,{\rm pass} \mid \rho_I, \mathcal{M}_b,\mathcal{M}_F) & = \sfrac{1}{9},\, \text{ and }\\
\label{eq:3box01}
p(\neg b,{\rm pass}\mid \rho_I, \mathcal{M}_b,\mathcal{M}_F) & = 0.
\end{align}

Finally, we can now apply these to the ABL rule in Eq.~\eqref{eq:ABL_rule}. If we assume that there is only one particle that can be in any one of the boxes, the ABL rule predicts that $p(b=1 \vert \rho_I,\rho_F,\mathcal{M}_1) = 1$, indicating that if the final state is $\rho_F$, the particle would have been in box $b=1$ were we to implement the intermediate measurement $\mathcal{M}_1$. However, the same conclusion holds if, instead, we implement $\mathcal{M}_2$ and ask whether the particle was in box $b=2$, since $p(b=2 \vert \rho_I,\rho_F,\mathcal{M}_2) = 1$ as well. This induces us to a puzzling interpretation that a single particle can be found in more than one box with certainty.

Before examining the (metaphysical) assumptions made above we note that we have considered what is called a \emph{logical} (i.e. possibilistic) form of the paradox~\cite{LeiferSpekkensontology}, since the ABL rule equals 1 for both measurements. We could instead relax this assumption and observe that whenever $$p(b=1 \mid \rho_I,\rho_F,\mathcal{M}_1) + p(b=2 \mid \rho_I,\rho_F,\mathcal{M}_2) > 1,$$ we still arrive at a paradoxical conclusion (now probabilistically, rather than possibilistically) because it remains puzzling to reason classically in the same way as before~\cite{Maroney1,MARONEY2}. In particular, we will say that the $N$-box PPS scenario leads to a non-logical (i.e. probabilistic) paradox whenever the ABL statistical predictions violate the inequality
\begin{equation}
\label{eq:probNbox}
\sum_{b=1}^{N}p(b \mid \rho_I,\rho_F,\mathcal{M}_b)\leq 1.
\end{equation}
The logical $N$-box paradox introduced in Ref.~\cite{albert1985curious} arises as an extreme case of violation of Ineq.~\eqref{eq:probNbox}. 

\subsection{Adjudicating the operational assumptions underlying  the $N$-box paradox}

The preceding discussion naturally divides into two distinct aspects. On the one hand, the scenario itself, together with the ABL rule, provides a straightforward \emph{operational} description of the quantum statistics. Although this description is \emph{not} theory-independent (since it relies on the operational predictions of quantum theory) it is nonetheless detached from any particular interpretation. The paradoxical conclusion, on the other hand, arises \emph{only} when we impose a series of strong interpretational assumptions. These assumptions merit significant scrutiny, and many have been investigated in detail in Refs.~\cite{LeiferSpekkensontology,LeiferspekkensPPScontextuality,PuseyLeifer2015,MARONEY2,ravon2007three,kirkpatrick2007reply}. A natural question is whether removing some of these underlying metaphysical assumptions invalidates the paradoxical reasoning. This question has been formally addressed through the introduction of toy models~\cite{kirkpatrick2003classical,LeiferSpekkensontology,MARONEY2,ravon2007three} and no-go results~\cite{LeiferspekkensPPScontextuality,PuseyLeifer2015}, which have advanced our understanding of this puzzle.

Arguably, the most notable output of this careful analysis, specifically from Refs.~\cite{MARONEY2,ravon2007three,LeiferSpekkensontology}, is that for the output ABL statistics to be genuinely paradoxical, in the sense of resisting simple classical toy models, one must impose operational constraints. These constraints are at least twofold: First, we are allowed to implement intermediate measurements to the system sequentially, before a post-selection takes place, to check whether there is actually more than one particle. A key constraint required for  the classicality bound in Ineq.~\eqref{eq:probNbox} to hold is therefore that we never find two particles; the system behaves as if a single particle is always present in one and only one box~\cite{MARONEY2}, described operationally by the constraint
\begin{equation}\label{eq:no_cheating_condition}
    p(b,b' \mid \rho_I, \mathcal{M}_b,\mathcal{M}_{b'}) = 0, \quad \forall b \neq b'.
\end{equation}

Second, one must carefully distinguish between disturbance that is visible in the post-selection statistics and disturbance that is hidden from them---the former is ruled out, while the latter remains possible~\cite{MARONEY2,LeiferspekkensPPScontextuality}. Since this distinction will not be central to what follows, we relegate the technical details to Refs.~\cite{LeiferSpekkensontology,MARONEY2,harrigan2010Einstein}. The key point, however, is that any intermediate measurement could in principle disturb the actual state of the system, provided that this disturbance remains undetectable in the final statistics. This is expressed by the following operational non-disturbance relation: 
\begin{align}
&p(b,{\rm pass} \mid \rho_I,\mathcal{M}_b,\mathcal{M}_F)+p(\neg b,{\rm pass} \mid \rho_I, \mathcal{M}_b,\mathcal{M}_F)  \nonumber \\
&=  p({\rm pass} \mid \rho_I,\mathcal{M}_F), \label{eq:op_non_dist} 
\end{align}
for all $b\in\{1,2,\ldots,N\}$. If the intermediate measurement disturbs the system in a way that later affects the post-selection statistics, then the operational non-disturbance condition is violated---and classical toy models can easily reproduce the apparent paradox~\cite{kirkpatrick2003classical,LeiferSpekkensontology,MARONEY2}.

In earlier formulations~\cite{aharonov1990properties,albert1985curious} the logical $N$-box scenario has been stated without the operational non-disturbance condition~\eqref{eq:op_non_dist}. This realization comes from Ref.~\cite{MARONEY2}.  For future referencing, we organize this discussion in a definition:  
\begin{definition}[$N$-box PPS scenario]\label{defNbox}
    An \emph{$N$-box PPS scenario} is defined by a pre-selected state $\rho_I$, a post-selected state $\rho_F$, and a box basis $\{\Pi_b\}_{b=1}^{N+1}$ (with $\Pi_b = \ket{b}\bra{b}$) for the intermediate measurements. The intermediate measurement $\mathcal{M}_b = \{\Pi_b, \mathbb{1} - \Pi_b\}$ is performed for $b \in \{1,\dots,N\}$, while the measurement $\mathcal{M}_{N+1}$ is excluded. Furthermore, the statistics must satisfy the operational non-disturbing condition~\eqref{eq:op_non_dist}, that $p({\rm pass} \mid \rho_I,\mathcal{M}_F) \neq 0$, and the operational constraint from Eq.~\eqref{eq:no_cheating_condition}. We say that the ABL rule leads to a \emph{paradox} in this scenario whenever Ineq.~\eqref{eq:probNbox} is violated. Additionally, the paradox is said to be  \emph{logical} whenever the ABL rule attains only $0$ or $1$ values.  
\end{definition}

We must note that the requirement $p({\rm pass}\mid\rho_I,\mathcal{M}_F)=\mathrm{Tr}[\rho_F\rho_I]\neq0$ is an extra physical restriction on top of the $N$-box scenario, serving as a precondition for the scenario to be well-posed. Post-selection, by definition, is conditioning on a successful outcome, and a scenario in which that outcome never occurs supplies no statistics to condition on. 

Requiring $\mathrm{Tr}[\rho_I\rho_F] \neq 0$ is also consistent with imposing the operational non-disturbance condition. This is already evident at the level of a single measurement $\mathcal{M}_b$. The ABL rule \eqref{eq:ABL_rule} for an outcome $b$ is the ordinary conditional probability $p(b,{\rm pass})/[p(b,{\rm pass})+p(\neg b,{\rm pass})]$ that follows from Bayes' theorem, which is meaningful only when its denominator is nonzero. The non-disturbance condition~\eqref{eq:op_non_dist} forces this denominator to be equal to $p({\rm pass}\mid\rho_I,\mathcal{M}_F)$. This in turn implies that once Eq.~\eqref{eq:op_non_dist} holds, the ABL rule is undefined for $\mathrm{Tr}[\rho_F\rho_I]=0$. In turn, there is no ABL statistic, paradoxical or otherwise, for Ineq.~\eqref{eq:probNbox} to be evaluated against.

\section{Quantum coherence}\label{sec:coherence}

A state $\rho$ is coherent with respect to a basis $A = \{\vert a \rangle\}_a$ when its density matrix has nonzero off-diagonal elements in that basis, i.e. when $\langle a \vert \rho  \vert a' \rangle \neq 0$ for some pair $a\neq a'$, and incoherent otherwise~\cite{baumgratz2014quantifying,aberg2006quantifyingsuperposition}. Quantum coherence is a well-studied quantifiable resource~\cite{streltsov2017colloquium,wu2021experimental,chitambar2019quantum} underlying applications in quantum computation~\cite{shor1999polynomial,ahnefeld2022role,naseri2022entanglement} and metrology~\cite{ahnefeld2026coherence}, as well as core discussions in quantum foundations~\cite{catani2023whyinterference,catani2023aspects}.

Put differently, coherence of $\rho$ relative to the basis $A=\{\vert a \rangle\}_a$ is the operational signature of the fact that any observable with eigenbasis $A$ fails to commute with $\rho$. This basis-dependent notion can be extended to a basis-independent one, defined as a \emph{collective property} of a set of states $\{\rho_1, \ldots,\rho_n\} \subseteq \mathcal{D}(\mathcal{H})$: the whole set is coherent whenever at least one pair of its members fails to commute. This extension was proposed recently by Designolle \emph{et al.}~\cite{designolle2021set}, who termed it \emph{set coherence}.

That basis-dependent coherence of $\rho_I$ and $\rho_F$ relative to the basis $B = \{\vert b \rangle\}_b$ matters for the paradox is, by construction, immediate: in the paradigmatic example the paradox is manifestly a consequence of these states being coherent in a specific way. What is non-trivial is whether the paradox itself---as a purely statistical, operational prediction---already \emph{certifies} coherence, given only that the underlying theory is quantum theory. Concretely, given the ABL statistics of a scenario satisfying the constraints in Def.~\ref{defNbox} above, does a violation of Ineq.~\eqref{eq:probNbox} force the relevant states and effects to be coherent, with no further assumption about them? This is precisely the type of question addressed by no-go results in quantum foundations, and an affirmative answer would show that no coherence-free subtheory of quantum theory can reproduce the statistical predictions of an $N$-box PPS paradox. This is precisely the content of our Theorem~\ref{thm:nonlogical} below (and also of Theorem~\ref{thm:pigeon-nonlogical} for the pigeonhole PPS scenario).

For the \emph{logical} case Pusey and Leifer~\cite{PuseyLeifer2015} showed that a paradox is possible iff the two states $\rho_I = \vert \psi_I \rangle \langle \psi_I \vert $ and $\rho_F = \vert \psi_F \rangle \langle \psi_F \vert$ are non-orthogonal, i.e.\ their overlap satisfies $0<\vert \langle \psi_I \vert \psi_F \rangle \vert^2 \leq 1$. For pure states this condition---assuming~\footnote{Our Theorem~\ref{thm:nonlogical} and Corollary~\ref{coro:anomalous} will complement this result by showing that whenever the states are parallel, i.e. $\vert \langle \psi_I \vert \psi_F \rangle \vert^2 = 1$, one cannot have a paradox.} both states are non-parallel---is equivalent to the statement that $\{\rho_I,\rho_F\}$ is set coherent~\cite{galvao2020quantum,oszmaniec2024measuring,wagner2024inequalities}, and hence, in particular, coherent relative to the basis $B$. This is also consistent with the fact that these paradoxes are proofs of contextuality, that contextuality requires (some form of) set coherence to be present~\cite{wagner2024coherence,wagner2024inequalities,zhang2026reassessing,catani2023aspects,hofmann2020contextuality,ming2023characterization}---though not conversely~\cite{hardy1999disentangling,catani2023whyinterference,spekkens2007evidence,bartlet2012reconstruction}---and that coherence is therefore a necessary resource for such paradoxes as well. 

Refs.~\cite{aharonov2016weakvalues,waegell2023negativemass,hofmann2009resolutionquantumparadoxesweak,hance2023contextuality,correa2021apparent} have also shown that typical logical PPS paradoxes are closely related to another non-classical phenomenon: \emph{anomalous weak values}~\cite{dressel2014colloquium,tamir2013introduction}. For an observable $A$ with pre- and post-selected states $\rho_{I}$ and $\rho_{F}$, the weak value is
 \begin{equation}
\label{eq:wv}
\langle A\rangle _{w}:=\frac{\mathrm{Tr}[\rho_{F}A\rho_{I}]}{\mathrm{Tr}[\rho_{F}\rho_{I}]}.
\end{equation}
This definition of a weak value is a generalized version valid for arbitrary mixed states from Refs.~\cite{dressel2015weak,hall2004prior,dziewior2019universality,WG23} that recovers the original formulation from Ref.~\cite{aharonov1988result} whenever the pre- and post-selected states are pure. Assuming that $A = \sum_{\lambda} \lambda \vert \lambda \rangle \langle \lambda \vert  $ we can re-write Eq.~\eqref{eq:wv} as
\begin{equation}
\label{eq:wv_bargmanns}
\langle A\rangle _{w}=\sum_\lambda \lambda \,\frac{\mathrm{Tr}[\rho_{F}\vert \lambda \rangle \langle \lambda \vert \rho_{I}]}{\mathrm{Tr}[\rho_{F}\rho_{I}]} = \sum_\lambda \lambda \frac{\Delta(\rho_F, \Pi_\lambda, \rho_I)}{\Delta(\rho_F,\rho_I)}.
\end{equation}
where $\Delta(\rho_1,\ldots,\rho_n)\equiv \mathrm{Tr}[\rho_1 \cdots \rho_n]$ is a \emph{Bargmann invariant}~\cite{wagner2025coherenceandcontextuality,oszmaniec2024measuring,wagner2026bargmannscenarios,bargmann1964note,zhang2026surveybargmanninvariantsgeometric,zhang2025geometrysets,pratapsi2025elementarycharacterizationbargmanninvariants} of the tuple of states $(\rho_1,\ldots,\rho_n)$. The number of states $n$ denotes the \emph{order} of the invariant. 

These invariants are known to characterize all unitary-invariant properties of any set of quantum states~\cite{oszmaniec2024measuring,wigderson2019mathematics}, and their values are signatures of set coherence~\cite{wagner2024inequalities,wagner2024quantumcircuits,giordani2021witnesses,galvao2020quantum,giordani2023experimental,jones2023distinguishability,hance2024counterfactuality} and its variants~\cite{maosheng2026multistate,wagner2024certifying,zamora2025semi,xu2026bargmanninvariantsgaussianstates,fernandes2024unitary,zhang2024local,grassl1998computing}. More specifically, if a Bargmann invariant of the form $\Delta(\rho_I, \Pi_\lambda,\rho_F)$ has negative or nonzero imaginary part, then $\{\rho_I,\rho_F,\Pi_\lambda\}$ is set coherent. In particular, this implies that either $\rho_F$ or $\rho_I$ must be coherent relative to the basis $\{\Pi_\lambda\}_\lambda$. Thus, the values of a Bargmann invariant can straightforwardly witness coherence in the ordinary, basis-dependent sense, even though it captures coherence in the broader sense of  Ref.~\cite{designolle2021set}. In what follows, we use the values of  Bargmann invariants as a witness of coherence relative to the basis of the intermediate measurements in our PPS scenarios.

We say that the weak value $\langle A \rangle_w$ is \emph{anomalous} whenever it is outside of the interval $[\lambda_{\rm min},\lambda_{\rm max}] \subset \mathbbm{R}$ provided by the smallest and the largest eigenvalues $\lambda$ of the observable $A$. Anomalous weak values are known to require coherence, via the nonclassicality of the underlying Bargmann invariants: if either $\rho_{I}$ or $\rho_{F}$ is incoherent in the eigenbasis $\{\Pi_{\lambda}\}$, the eigenprojector weak values $\langle\Pi_{\lambda}\rangle_{w}$ form a genuine probability distribution~\cite{lund2010measuring,higgins2015using,degosson2012weak,WG23}, with $0\leq \langle\Pi_{\lambda}\rangle_{w}\leq 1$ and $\sum_{\lambda}\langle\Pi_{\lambda}\rangle_{w}=1$, so that $\langle A\rangle_{w}$ is non-anomalous~\cite{WG23}. Any anomaly therefore implies that there is at least some third-order Bargmann invariant appearing in the numerator of Eq.~\eqref{eq:wv_bargmanns} that is inconsistent with all underlying states being incoherent with respect to a common basis, since such invariants must then be negative or have nonzero imaginary part. 

Together with the findings of Refs.~\cite{aharonov2016weakvalues,waegell2023negativemass,hofmann2009resolutionquantumparadoxesweak,hance2023contextuality,correa2021apparent}, this connection between weak values, Bargmann invariants, and coherence via the relation
\begin{align*}
    \text{logical PPS}\Rightarrow \text{weak value anomaly}\Rightarrow \\ \exists \lambda: \Delta(\rho_I,\Pi_\lambda,\rho_F)\notin [0,1] \Rightarrow \text{coherence},
\end{align*}
makes it transparent that logical PPS paradoxes require  coherence not only through the second-order invariants $\mathrm{Tr}(\psi_I \psi_F)$ identified by Ref.~\cite{PuseyLeifer2015}, but through third-order invariants as well.

If instead the two pre- and post-selected states are allowed to be arbitrary \emph{mixed} states, non-orthogonality is captured by the two-state overlap $\mathrm{Tr}(\rho_I \rho_F)$ rather than by $\vert\langle\psi_I\vert\psi_F\rangle\vert^2$. In this case, even granting the prior information that $\rho_I \neq \rho_F$, the condition $0< \mathrm{Tr}(\rho_I \rho_F) < 1$ is no longer sufficient to guarantee set coherence of the pre- and post-selected states. Moreover, we would ideally like the case $\rho_I = \rho_F$ to remain well-defined whenever $\mathrm{Tr}(\rho_I\rho_F) \neq 0$, i.e. we would like our analysis to require \emph{no prior} on the states involved---not on their purity, not on whether they are distinct, and not even on the underlying Hilbert space. 

Furthermore, to the best of our knowledge, no connection has yet been drawn between anomalous weak values for mixed states, as in Eq.~\eqref{eq:wv}, and non-logical PPS paradoxes. Together, these observations leave open the possibility that some incoherent subtheory of quantum theory could satisfy the operational desiderata of an $N$-box PPS scenario while still reproducing some form of the paradox. In what follows we rule this out: we provide no-go results that both recover the logical case and strengthen the connection between PPS paradoxes and the negativity of the third-order Bargmann invariants underlying weak values.

\section{$N$-box paradoxes and coherence}\label{sec:results}

If we assume that either $\rho_I$ or $\rho_F$ are incoherent states relative to $B = \{\Pi_b\}_b$, then it is simple to see that they must satisfy the non-disturbance condition. This is so because, if we write  $\rho_I=\sum_{b'=1}^{N+1} p_{b'} \Pi_{b'}$ then it follows that
\begin{equation}
    \Pi_b \rho_I \Pi_b = \sum_{b'}p_{b'}\Pi_b \delta_{bb'} = p_b \Pi_b
\end{equation}
and also
\begin{equation}
    \Pi_b^\perp \rho_I \Pi_b^\perp = \sum_{b' \neq b}p_{b'}\Pi_{b'}.
\end{equation}
Therefore, 
\begin{equation}
    \Pi_b \rho_I \Pi_b + \Pi_b^\perp \rho_I \Pi_b^\perp  = \rho_I,
\end{equation}
and we recover the operational non-disturbance condition by multiplying both sides of this equation by $\rho_F$ and taking the trace.

With that, we can start proving our no-go result by re-writing $p(b,{\rm pass} \mid \rho_I,\mathcal{M}_b, \mathcal{M}_F)$ under the assumption that these must satisfy the  operational non-disturbance constraint from Eq.~\eqref{eq:op_non_dist}. Note that this condition can be written as
\begin{align}
&\mathrm{Tr}[\rho_{F}\rho_{I}] = \mathrm{Tr}[\rho_{F}\Pi_{b}\rho_{I}\Pi_{b}] + \mathrm{Tr}[\rho_{F}(\mathbb{1}-\Pi_{b})\rho_{I}(\mathbb{1}-\Pi_{b})]\Rightarrow\nonumber\\
&2\mathrm{Tr}[\rho_{F}\Pi_{b}\rho_{I}\Pi_{b}]=  \mathrm{Tr}[\rho_F\rho_I \Pi_b]+\mathrm{Tr}[\rho_I\rho_F \Pi_b]\Rightarrow\nonumber\\
\label{eq:quant_non_dist_alt}
&\mathrm{Tr}[\rho_{F}\Pi_{b}\rho_{I}\Pi_{b}] = \frac{\mathrm{Tr}[\rho_{F}\{\Pi_{b},\rho_{I}\}]}{2}.
\end{align}
Provided with this relation, we have then that we can write
\begin{align}
p(b,{\rm pass} \,\vert\, \rho_I,\mathcal{M}_b,\mathcal{M}_F)  &\stackrel{\eqref{eq:Nbox11}}{=} \mathrm{Tr}[\rho_{F}\Pi_{b}\rho_{I}\Pi_{b}]\\\label{eq:quant_non_dist_alt1}
& \stackrel{\eqref{eq:quant_non_dist_alt}}{=}\frac{\mathrm{Tr}[\rho_{I}\rho_{F}\Pi_{b}]+\mathrm{Tr}[\Pi_{b}\rho_{F}\rho_{I}]}{2}\\
\label{eq:quant_non_dist_alt2}
& = \frac{\mathrm{Tr}[(\Pi_{b}\rho_{F}\rho_{I})^{\dagger}]+\mathrm{Tr}[\Pi_{b}\rho_{F}\rho_{I}]}{2}\\
\label{eq:quant_non_dist_alt3}
& = \frac{\mathrm{Tr}[\Pi_{b}\rho_{F}\rho_{I}]^{*}+\mathrm{Tr}[\Pi_{b}\rho_{F}\rho_{I}]}{2}\\
\label{eq:quant_non_dist_alt4}
& = \mathrm{Re}\left(\mathrm{Tr}[\Pi_{b}\rho_{F}\rho_{I}] \right)
\end{align}
where $\{A,B\}=AB+BA$ is the anti-commutator of $A$ and $B$. To obtain Eq.~\eqref{eq:quant_non_dist_alt2} from Eq.~\eqref{eq:quant_non_dist_alt1} we use the fact that $\Pi_{j},\rho_{F}$ and $\rho_{I}$ are Hermitian. Eq.~\eqref{eq:quant_non_dist_alt3} is obtained from Eq.~\eqref{eq:quant_non_dist_alt2} using $\mathrm{Tr}[M^{\dagger}]=\mathrm{Tr}[M]^{*}$ for any matrix $M$.

Using the above we can prove the following no-go theorem:

\begin{theorem}[A coherence-free subtheory cannot yield an $N$-box PPS paradox]\label{thm:nonlogical}
For the $N$-box scenario in Def.~\ref{defNbox}, if either the pre- or the post-selected state is incoherent with respect to the intermediate measurement basis $B$, then Ineq.~\eqref{eq:probNbox} is satisfied.
\end{theorem}
\begin{proof}
Using Eq.~\eqref{eq:quant_non_dist_alt4} we can write 
\begin{align}
\label{eq:quantnon-dist_N_Box0}
\sum_{b=1}^{N}\mathrm{Tr}[\rho_{F}\Pi_{b}\rho_{I}\Pi_{b}] & \stackrel{\eqref{eq:quant_non_dist_alt4}}{=} \mathrm{Re}\left(\mathrm{Tr}\left[\sum_{b=1}^{N}\Pi_{b}\rho_{F}\rho_{I}\right]\right),\\
\label{eq:quantnon-dist_N_Box1}
& = \mathrm{Re}\left(\mathrm{Tr}\left[(\mathbb{1}-\Pi_{N+1})\rho_{F}\rho_{I}\right]\right),\\
\label{eq:quantnon-dist_N_Box2}
& = \mathrm{Tr}[\rho_{F}\rho_{I}]-\mathrm{Re}\left(\mathrm{Tr}[\Pi_{N+1}\rho_{F}\rho_{I}]\right),\\
\label{eq:quantnon-dist_N_Box3}
& = \mathrm{Tr}[\rho_{F}\rho_{I}]\L1-\mathrm{Re}\L\langle \Pi_{N+1}\rangle _{w}\R\R.
\end{align}
To get Eq.~\eqref{eq:quantnon-dist_N_Box3} from Eq.~\eqref{eq:quantnon-dist_N_Box2} we use Eq.~\eqref{eq:wv}, thus obtaining the relevant weak value for our effect. Recall that, by assumption on the scenario, we assume that $\mathrm{Tr}[\rho_F\rho_I] \neq 0$ so that we can consistently apply the ABL rule while simultaneously assume the validity of the operational non-disturbance condition~\eqref{eq:op_non_dist}.

Using the results from Ref.~\cite{WG23}, if either $\rho_{F}$ or $\rho_{I}$ is incoherent with respect to the  basis $B = \{\Pi_b\}_{b=1}^{N+1}$, then the weak values $\langle \Pi_{b} \rangle _{w}$ form a valid probability distribution (note that, here, $b$ goes up to $N+1$). Consequently, 
\begin{equation}
\label{bergamannbound}
 0\leq \langle \Pi_{N+1}\rangle_{w} \leq 1
\end{equation}
which can be used in Eq.~\eqref{eq:quantnon-dist_N_Box3} to obtain
\begin{align}
\sum_{b=1}^{N}p(b,{\rm pass}\mid \rho_I,\mathcal{M}_b,\mathcal{M}_F)&\leq \mathrm{Tr}[\rho_{F}\rho_{I}]\nonumber \\&=p({\rm pass}\mid \rho_I,\mathcal{M}_F).\label{eq:probNbox1}
\end{align}
Using the operational non-disturbance condition given by Eq.~\eqref{eq:op_non_dist} for each $b\in \{1,\ldots,N\}$ and the ABL rule of conditional probability given by
Eq.~\eqref{eq:ABL_rule}, we get
\begin{align}
&\sum_{b=1}^{N}p(b\mid \rho_I,\rho_F,\mathcal{M}_b)\stackrel{\eqref{eq:ABL_rule}}{=} \nonumber \\
&\sum_{b=1}^{N}\frac{p(b,{\rm pass}\mid \rho_I,\mathcal{M}_b,\mathcal{M}_F)}{p(b,{\rm pass}\mid \rho_I,\mathcal{M}_b,\mathcal{M}_F)+p(\neg b,{\rm pass}\mid \rho_I,\mathcal{M}_b,\mathcal{M}_F)}\nonumber \\
&\stackrel{\eqref{eq:op_non_dist}}{=}\sum_{b=1}^{N}\frac{p(b,{\rm pass}\mid \rho_I,\mathcal{M}_b,\mathcal{M}_F)}{p({\rm pass}\mid \rho_I,\mathcal{M}_F)} \stackrel{\eqref{eq:probNbox1}}{\leq} \frac{p({\rm pass}\mid \rho_I,\mathcal{M}_F)}{p({\rm pass}\mid \rho_I,\mathcal{M}_F)} \nonumber \\
&= 1.
\end{align}
This completes the proof.
\end{proof}

The relation \eqref{eq:quantnon-dist_N_Box3} admits a transparent interpretation, one that lets us view this result as an extension of earlier connections between anomalous weak values and the presence of a paradox in the $N$-box PPS scenario~\cite{aharonov2016weakvalues}. In particular, we now show that it is precisely the coherence witnessed by the \emph{negativity} of a third-order Bargmann invariant that quantifies the violation of Ineq.~\eqref{eq:probNbox}.

\begin{corollary}[Negative weak values of the excluded box quantify the non-logical PPS paradox]
\label{coro:anomalous}
For the $N$-box scenario in Def.~\ref{defNbox}, the sum of ABL statistics satisfy
\begin{equation}\label{eq:sumweakvalue}
\sum_{b=1}^{N}p(b\,\vert\,\rho_I,\rho_F,\mathcal{M}_b) = 1-\mathrm{Re}\big(\langle \Pi_{N+1}\rangle_{w}\big).
\end{equation}
Consequently, a paradox  occurs if and only if the weak value of the excluded projector $\Pi_{N+1}$ has negative real part,
\begin{equation}
\mathrm{Re}\big((\Pi_{N+1})_{w}\big) < 0,
\end{equation}
that is, if and only if $\langle \Pi_{N+1} \rangle _{w}$ is anomalous. 
\end{corollary}

\begin{proof}
Starting from Eq.~\eqref{eq:quantnon-dist_N_Box3}, established in the proof of \cref{thm:nonlogical},
\begin{equation}
\sum_{b=1}^{N}\mathrm{Tr}[\rho_{F}\Pi_{b}\rho_{I}\Pi_{b}] = \mathrm{Tr}[\rho_{F}\rho_{I}]\L1-\mathrm{Re}\big(\langle \Pi_{N+1}\rangle_{w}\big)\R,
\end{equation}
we divide both sides by $\mathrm{Tr}[\rho_{F}\rho_{I}]=p({\rm pass}\mid\rho_I,\mathcal{M}_F)$, which is nonzero by assumption. Applying, for each $b\in\{1,\ldots,N\}$, the operational non-disturbance condition \eqref{eq:op_non_dist} together with the ABL rule \eqref{eq:ABL_rule}---exactly as in the proof of \cref{thm:nonlogical}---turns the left-hand side into $\sum_{b=1}^{N}p(b\mid\rho_I,\rho_F,\mathcal{M}_b)$, yielding the identity~\eqref{eq:sumweakvalue}. From that, the left-hand side of Eq.~\eqref{eq:sumweakvalue} exceeds $1$ if and only if $\mathrm{Re}\big(\langle \Pi_{N+1}\rangle_{w}\big)<0$, ie., the weak value is anomalous. 
\end{proof}

It is simple to see that whenever the pre- and post-selected states satisfy that $p({\rm pass}\mid \rho_I,\mathcal{M}_F) = \mathrm{Tr}[\rho_I\rho_F] = 1$ we have that 
\begin{equation*}
    \langle \Pi_{N+1}\rangle_w = \frac{\mathrm{Tr}[\rho_F \Pi_{N+1}\rho_I]}{\mathrm{Tr}[\rho_F\rho_I]} =\langle \psi \vert \Pi_{N+1} \vert \psi \rangle \geq 0
\end{equation*}
where we have used that $\mathrm{Tr}[\rho_I\rho_F]=1 \Rightarrow \rho_I = \rho_F = \psi$~\cite{vaidman1996weak}. Since this weak value is nonnegative, we get from Corollary~\ref{coro:anomalous} that a paradox is impossible. 

We can now see a few explicit examples of how the negativity of this particular weak value relates to anomaly, in this quantifiable manner. We start with the  logical three-box paradox reviewed earlier.

\begin{figure}[t]
\centering
\begin{tikzpicture}[>=Latex, font=\small]
\fill[blue!8] (0,-0.14) rectangle (3,0.14);
\fill[red!12] (3,-0.14) rectangle (6,0.14);
\node[font=\footnotesize] at (1.4,0.46) {Classical (no paradox)};
\node[font=\footnotesize] at (4.5,0.46) {Quantum (paradox)};
\draw[dashed] (3,-0.14)--(3,1);
\node[above, font=\footnotesize, align=center] at (3,0.98)
  {$p(b=1 \mid \rho_I,\rho_F,\mathcal{M}_b)+p(b=2 \mid \rho_I,\rho_F,\mathcal{M}_b)=1$};
\draw[->,thick] (-0.25,0)--(6.6,0);
\foreach \x/\s/\r in {0/0/1, 3/1/0, 6/2/{-1}}{
  \draw (\x,0.14)--(\x,-0.14);
  \node[font=\footnotesize] at (\x,-0.35) {$\s$};
  \node[font=\footnotesize, gray] at (\x,-0.72) {$\r$};
}
\node[left, font=\footnotesize] at (-0.2,-0.35) {$\Sigma$:};
\node[left, font=\footnotesize, gray] at (-0.2,-0.72) {Re:};
\fill (6,0) circle (1.8pt);
\node[font=\footnotesize] at (6,-1.08) {Optimal};
\end{tikzpicture}
\caption{\textbf{The 3-box paradox as a function of $\mathrm{Re}\left(\langle \Pi_3\rangle_w\right)$.} Under the operational non-disturbance condition~\eqref{eq:op_non_dist}, the $3$-box statistics obey $\Sigma\equiv\sum_{b=1}^{2}p(b\mid \rho_I,\rho_F,\mathcal{M}_b)=1-\mathrm{Re}\big(\langle \Pi_{3}\rangle_w\big)$ (see~\cref{coro:anomalous}); the lower (grey) axis shows the corresponding $\mathrm{Re}\big(\langle \Pi_{3}\rangle _w\big)$. The classicality bound $\Sigma\le1$, Ineq.~\eqref{eq:probNbox}, is violated  in the pink shaded ``paradox'' region, where the excluded box has an anomalous weak value $\mathrm{Re}\big(\langle \Pi_{3}\rangle_w\big)<0$. The logical three-box paradox saturates the algebraic maximum $\Sigma=2$, i.e.\ $\langle \Pi_3\rangle_w=-1$.}
\label{fig:bound}
\end{figure}
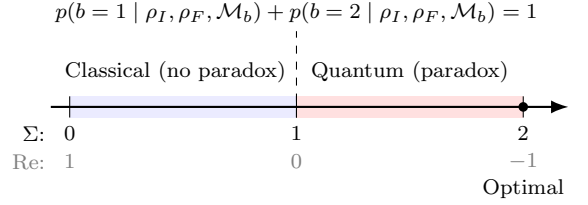

\begin{exa}[Weak values and the three-box paradox]
\label{ex:threebox}
For the paradigmatic three-box paradox ($N=2$), with $\ket{\psi_{I}}=\tfrac{1}{\sqrt{3}}(\ket{1}+\ket{2}+\ket{3})$, $\ket{\psi_{F}}=\tfrac{1}{\sqrt{3}}(\ket{1}+\ket{2}-\ket{3})$ and excluded box $\Pi_{3}=\ketbra{3}{3}$, one has $\mathrm{Tr}(\rho_I\rho_F) = \vert \braket{\psi_{F}}{\psi_{I}}\vert^2=\sfrac{1}{9}$, and
\begin{align}
\langle \Pi_{3} \rangle_{w}=\frac{\mathrm{Tr}[\rho_F\Pi_{3}\rho_{I}]}{\mathrm{Tr}[\rho_{F}\rho_{I}]}
&=\frac{\braket{\psi_I}{\psi_F}\braket{\psi_F}{3}\braket{3}{\psi_I}}{|\braket{\psi_{F}}{\psi_{I}}|^{2}}\nonumber\\
&=\frac{\big(\tfrac{1}{3}\big)\big(-\tfrac{1}{\sqrt{3}}\big)\big(\tfrac{1}{\sqrt{3}}\big)}{\sfrac{1}{9}}=-1.
\end{align}
By Eq.~\eqref{eq:sumweakvalue}, $\sum_{b=1}^{2}p(b \mid \psi_I,\psi_F,\mathcal{M}_b)=1-(-1)=2$, the maximal logical paradox, the same result we have found earlier in Sec.~\ref{sec:scenario}.  The excluded box thus carries the extreme anomalous weak value $-1$, in agreement with Eqs.~\eqref{eq:3box11}--\eqref{eq:3null}.
\end{exa}

This example complements the analysis from Ref.~\cite{waegell2023negativemass}, solidifying the role of weak value anomaly in the $N$-box scenario. We can also see a different case of a family of paradoxes which happen at the non-logical level, yet are also quantified by the degree of violation relative to the negativity of the weak value $\langle \Pi_{N+1} \rangle_w$.

\begin{exa}[A family of non-logical three-box paradoxes]
We now construct an example of a paradox with mixed pre- and post-selected states. Let
\begin{align}
\rho_I&=p_i\,\ket{\psi_I}\bra{\psi_I}\,\,+(1-p_i)\,\ket{3}\bra{3},\\
\rho_F&=p_f\ket{\psi_F}\bra{\psi_F}+(1-p_f)\ket{3}\bra{3},
\end{align}
where $\ket{\psi_I}$ and $\ket{\psi_F}$ are taken to be same as in~\cref{ex:threebox} and $0\leq  p_i,p_f<  1$.
The intermediate measurements are taken to be the computational basis projectors $\Pi_1 \equiv \vert 1 \rangle \langle 1 \vert$ and $\Pi_2 \equiv \vert 2 \rangle \langle 2 \vert$. For the above choice, using Eqs.~\eqref{eq:Nbox11},~\eqref{eq:Nbox01}~and \eqref{eq:null} we have 
\begin{align*}
p(b, {\rm pass} \mid \rho_I, \mathcal{M}_b,\mathcal{M}_F)& =  \frac{p_{i}p_{f}}{9},\\
p(\neg b,{\rm pass} \mid \rho_I, \mathcal{M}_b, \mathcal{M}_F)&=\frac{p_{f}p_{i}+3-2(p_{f}+p_{i})}{3},\\
p({\rm pass} \mid \rho_I, \mathcal{M}_F)&= \frac{(3-2p_{f})(3-2p_{i})}{9}.
\end{align*}
for $b\in\{1,2\}$. It can be checked that this statistics satisfy the operational non-disturbance condition given by Eq.~\eqref{eq:op_non_dist} for all values of $p_{i}$ and $p_{f}$. One can obtain the ABL probabilities as 
\begin{equation}
p(b \mid \rho_I, \rho_F, \mathcal{M}_b) = \frac{p_{i}p_{f}}{(3-2p_{i})(3-2p_{f})}.
\end{equation}
Furthermore, the weak value $\langle \Pi_{3} \rangle _{w}$ is given by 
\begin{equation}
\langle \Pi_{3}\rangle _{w}=\frac{9+2p_{f}p_{i}-6(p_{f}+p_{i})}{(3-2p_{f})(3-2p_{i})}.
\end{equation}
The paradox appears when $\text{Re} \left(\langle \Pi_{3}\rangle_{w}\right)<0$. In this case, the region of $p_{i}$ and $p_{f}$ which satisfies this criterion is given by
\begin{equation}
6(p_{i}+p_{f})-2p_{i}p_{f}>9.
\end{equation}
For all those values one has a non-logical instance of a three-box paradox.
\end{exa}

While the violation of Ineq.~\eqref{eq:probNbox}, under the assumptions in Def.~\ref{defNbox}, implies coherence, the converse does not hold. It is possible that coherent quantum states satisfy all the requirements of the scenario and nevertheless yield no paradox. 

\begin{exa}[Coherent states which yield no paradox]
    Let us consider pre- and post-selection according to the quantum states
    \begin{align}
        \rho_I(\theta)&=\frac{1}{3}\mathbb{1}_{3}+\frac{\cos(\theta)}{6}\left(\vert 1 \rangle \langle 2 \vert + \vert 2 \rangle \langle 1 \vert \right),\\ \rho_F(\theta)&=\frac{1}{3}\mathbb{1}_{3}+\frac{\sin(\theta)}{6}\left(i\vert 1 \rangle \langle 2 \vert - i\vert 2 \rangle \langle 1 \vert \right),
    \end{align}
    with $\theta \in [0,2\pi)$ and $\mathbb{1}_{3} = \sum_{b=1}^3 \vert b \rangle \langle b \vert$. Their overlap is equal to $\mathrm{Tr}[\rho_I\rho_F]=\sfrac{1}{3}$ for all $\theta$, and they satisfy the operational non-disturbance condition~\eqref{eq:op_non_dist} since
    \begin{align*}
        p(b,{\rm pass}\mid \rho_I,\mathcal{M}_b,\mathcal{M}_F) = \mathrm{Tr}[\rho_F \Pi_b \rho_I \Pi_b] = \frac{1}{9},\\
        p(\neg b,{\rm pass}\mid \rho_I,\mathcal{M}_b,\mathcal{M}_F) = \mathrm{Tr}[\rho_F \Pi_b^\perp \rho_I \Pi_b^\perp] = \frac{2}{9},
    \end{align*}
    for $b\in \{1,2\}$. Both of these states are coherent with respect to the intermediate measurement basis $\{\vert b \rangle\}_{b=1}^3$ since they both have non-zero off-diagonal terms $\langle 1 \vert \rho \vert 2 \rangle$. Nevertheless, they also do not allow for a paradox since
    \begin{equation}
        \sum_{b=1}^2p(b \mid \rho_I,\rho_F,\mathcal{M}_b) = \frac{\sfrac{1}{9}}{\sfrac{1}{3}}+\frac{\sfrac{1}{9}}{\sfrac{1}{3}} = \frac{2}{3} <1.
    \end{equation}
\end{exa}

We have thus shown that a violation of Ineq.~\eqref{eq:probNbox} implies that both pre- and post-selected states are coherent with respect to the intermediate measurement basis; it follows that the logical $N$-box paradox also requires quantum coherence. This conclusion, however, rests so far on the operational non-disturbance constraint of Eq.~\eqref{eq:op_non_dist}. As we discuss in Sec.~\ref{sec:LG}, this is not incidental: Maroney~\cite{Maroney1,MARONEY2} showed that, under this constraint, Ineq.~\eqref{eq:probNbox} can be recast as a Leggett--Garg inequality, whose violation is itself well known to require some form of quantum coherence. We now proceed to see how the exact link we establish, via the weak value of the excluded box, fits with this previous result.

\subsection{Revisiting the connection with Leggett and Garg's notion of macrorealism}\label{sec:LG}

The two operational constraints identified in Sec.~\ref{sec:scenario}---Eq.~\eqref{eq:no_cheating_condition} and the operational non-disturbance condition of Eq.~\eqref{eq:op_non_dist}---are not \emph{ad hoc}: they correspond, respectively, to the two assumptions Maroney identifies as jointly constituting a classical, macrorealist account of the three-box paradox~\cite{Maroney1,MARONEY2}, in the sense of Leggett and Garg's notion of macrorealism~\cite{leggett1985quantum}. We now discuss how this correspondence, together with the recent work of Schmid~\cite{schmid2024reviewreformulation}, offers another perspective on why coherence-free subtheories cannot reproduce the operational predictions of the $N$-box scenario.

\emph{No-cheating and macrorealism per se.---}In the following we refer to the condition of Eq.~\eqref{eq:no_cheating_condition} as the \emph{no-cheating} condition, since it rules out the simple ``cheating'' classical toy model described by Maroney~\cite[Appendix A.4]{MARONEY2}. Maroney's own statement of macrorealism, in the three-box scenario, is that:
\begin{quote}
``\emph{In the context of the three box paradox, macrorealism is simply the claim that the ball is, at any time, in one box, and only in one box.}''~\cite{Maroney1}.
\end{quote}
This is precisely what the no-cheating condition~\eqref{eq:no_cheating_condition} enforces operationally, prior to post-selection: that sequential intermediate measurements never find the particle in two boxes at once.

Crucially, Maroney stresses that this assumption alone---sometimes called \emph{macrorealism per se}---is ontological (i.e., a commitment concerning the underlying theory that attempts to explain the data~\cite{harrigan2010Einstein}), not operational. This is why, in our treatment, we have deliberately avoided referring to Eq.~\eqref{eq:no_cheating_condition} as ``the macrorealism condition,'' reserving that term for the ontological commitment above. Equation~\eqref{eq:no_cheating_condition} is instead an operationalization of the requirement that one be able to test whether there are two particles present rather than one.

\emph{Non-disturbance and non-invasive measurability.---}The second assumption central to Leggett and Garg's notion of macrorealism, which converts macrorealism \emph{per se} into an empirically testable claim, is \emph{non-invasive measurability}: the requirement that it be possible, in principle, to determine which box the ball occupies without disturbing it, so that the statistics of a later measurement do not depend on whether (or how) an earlier box-check was performed.

This is exactly what the operational non-disturbance condition~\eqref{eq:op_non_dist} captures, and it is the assumption under which the classicality bound~\eqref{eq:probNbox} relates to a genuine Leggett--Garg inequality~\cite{leggett1985quantum}. This bound holds for any theory satisfying macrorealism \emph{per se} together with non-invasive measurability. 

A recent reformulation of macrorealism is due to Schmid~\cite{schmid2024reviewreformulation}, who argues that both macrorealism \emph{per se} and non-invasive measurability (together with the various technical difficulties that have accumulated around tests of Leggett--Garg inequalities) can be replaced by a single, cleaner criterion: a theory is macrorealist if and only if every macroscopic system is described by a system in a \emph{strictly classical}~\cite{schmid2021characterization,schmid2024reviewreformulation,selby2023accessible,schmid2025shadowssubsystemsof} generalized probabilistic theory (GPT)~\cite{plavala2023general}. GPTs are a broad framework for formalizing physical theories beyond quantum theory, of which quantum theory itself is a particular instance. Within quantum theory, a strictly classical fragment can be viewed as one captured by an incoherent subtheory of quantum theory relative to some orthonormal basis. 

Under this perspective, it becomes clear that the violation, by quantum theory, of the bound~\eqref{eq:probNbox} is precisely what \cref{thm:nonlogical} and \cref{coro:anomalous} trace back to coherence. In Schmid's~\cite{schmid2024reviewreformulation} sense---and assuming quantum theory as the relevant underlying theory---any violation of such a Leggett--Garg inequality witnesses the fact that an incoherent subtheory cannot reproduce these effects, thereby witnessing the coherence of the relevant states and effects.

\section{Quantum pigeonhole principle}\label{sec:pigeon}

We now extend our findings to a distinct PPS scenario, related to a paradox known as the \emph{quantum pigeonhole principle}~\cite{aharonov2016quantum}. The (classical) pigeonhole principle~\cite{rittaud2013pigeonhole} states the following~\cite{aharonov2016quantum}:
\begin{quote}
    ``\emph{If you put three pigeons in two pigeonholes, at least two of the pigeons end up in the same hole.}''
\end{quote}
This principle embeds a certain type of classical intuition regarding how one should \emph{count}, that can also be (seemingly) challenged by an instance of a PPS scenario. 

Aharonov \emph{et al.}~\cite{aharonov2016quantum} proposed the following PPS scenario where we consider a system of three qubits---interpreted as characterizing three perfectly distinguishable particles (``pigeons'') that can be put in two boxes (``pigeonholes''), so the full system lives in $\mathcal{H}=(\mathbbm{C}^2)^{\otimes3}$---each having a specified ``box'' basis, that we will take to be written as $\{\ket{L},\ket{R}\}$ for a ``left'' and ``right'' box.  The product basis $C=\{\ket{c_1c_2c_3}\}_{c_i\in\{L,R\}}$ (having $8$ elements) provides a fully separable orthonormal basis for the full space. 

For example, if we prepare the system to be initially in the state $\vert \psi_I \rangle = \vert L\rangle \vert L \rangle \vert L \rangle $ this is to be interpreted as the fact that all three particles are in the left box. A measurement of the form $\{\vert L \rangle \langle L \vert \otimes \mathbb{1}_{2 \otimes 2}, \vert R \rangle \langle R \vert \otimes \mathbb{1}_{2 \otimes 2}\}$ is then interpreted as a measurement that checks whether the first particle is in the left or in the right box. As before, these are strong commitments to a classical interpretation of what measurements do (which \emph{per se}, are known to be in tension with various celebrated no-go results~\cite{kochen1967problem,budroni2022kochenspecker,bell1966problem}) and should be taken here as a means to expose the kind of way this classical reasoning leads to paradoxical conclusions.

As before, we define a family of intermediate measurements which in this case will be related, naturally, to checking whether there are two particles in the same box, since we are trying to investigate the validity of the pigeonhole principle. Thus, for each unordered pair $\{i,j\}\subset\{1,2,3\}$, we define the ``same-box'' projector $\Pi_{ij}^{\rm same}$ as
\begin{equation}
\Pi^{\rm same}_{ij} = \ketbra{LL}{LL}_{ij}\otimes{\mathbb{1}_2} + \ketbra{RR}{RR}_{ij}\otimes{\mathbb{1}_2},
\end{equation}
where the identities apply in the subspace $k\in \{1,2,3\}\setminus \{i,j\}$ of the third particle. The dichotomic intermediate measurement $\mathcal{M}_{ij}=\{\Pi^{\rm same}_{ij},\mathbb{1}_8-\Pi^{{\rm same}}_{ij}\}$ is interpreted as asking the question ``\emph{are particles $i$ and $j$ in the same box?}''. Note here the close structure relative to the $N$-box paradox from Sec.~\ref{sec:scenario}. Writing ``${\rm same}_{ij}$'' and ``${\rm diff}_{ij}$'' for the two outcomes of $\mathcal{M}_{ij}$, the ABL rule~\eqref{eq:ABL_rule} is, for pre- and post-selected states $\rho_I,\rho_F$, given by 
\begin{equation}
\label{eq:pigeon-ABL}
p({\rm same}_{ij}\mid\rho_I,\rho_F,\mathcal{M}_{ij}) = \frac{p({\rm same}_{ij},{\rm pass})}{p({\rm same}_{ij},{\rm pass})+p({\rm diff}_{ij},{\rm pass})}.
\end{equation}
Above, we have denoted $$p({\rm same}_{ij},{\rm pass}) \equiv p({\rm same}_{ij},{\rm pass} \mid \rho_I, \mathcal{M}_{ij}, \mathcal{M}_F)$$ for brevity, as we did in Eq.~\eqref{eq:ABL_rule}.

We are now ready to provide an example of a PPS paradox for this scenario, which will be, as before, an instance of how the classical interpretational baggage, together with reasoning via the ABL rule, leads to a paradox. Let us consider for the pre-selection the state
\begin{equation}
\ket{\psi_I}=\ket{+}\ket{+}\ket{+}
\end{equation}
and for the post-selection
\begin{equation}
\ket{\psi_F}=\ket{+_i}\ket{+_i}\ket{+_i},
\end{equation}
where $\ket{+}=(\ket{L}+\ket{R})/\sqrt{2}$ and $\ket{+_i}=(\ket{L}+i\ket{R})/\sqrt{2}$. A direct computation gives $p({\rm diff}_{ij}\mid\rho_I,\rho_F,\mathcal{M}_{ij})=1$ for \emph{every} pair. This leaves us with the paradoxical counterfactual reasoning: were we to check whether there are two particles in the same box, for any pair of particles, we would find them certainly in different boxes. Yet classically, any definite assignment of each particle to $L$ or $R$ must place at least one pair together---hence the paradox.

Before adjudicating the minimal operational desiderata for a paradox to truly be present, it is moreover somewhat clear that there is a simple non-logical extension of the paradox. To see this, note that a definite classical assignment of the three particles to $L$ or $R$ produces exactly one coinciding pair if the assignment is non-uniform (e.g.\ $LLR$, six such cases out of eight) and all three coinciding pairs if it is uniform ($LLL$ or $RRR$, two cases out of eight); the number of coinciding pairs is thus never zero, but it \emph{is} exactly one for the generic (non-uniform) assignments. Averaging over any classical mixture of assignments therefore gives the tight bound
\begin{equation}
\label{eq:pigeon-ineq}
\sum_{\{i,j\}}p({\rm same}_{ij}\mid\rho_I,\rho_F,\mathcal{M}_{ij}) \geq 1,
\end{equation}
which cannot be strengthened, since the generic assignments already saturate it at exactly $1$. 

We say the ABL statistics exhibit a \emph{pigeonhole paradox} whenever they \emph{violate} Eq.~\eqref{eq:pigeon-ineq}. A \emph{logical} pigeonhole paradox is the extremal instance of this, where $p({\rm same}_{ij},{\rm pass})=0$ for all three pairs simultaneously---the optimal case for a paradox, in which the ABL rule predicts, with certainty, that \emph{no} pair will ever coincide. This is exactly the situation provided by Aharonov \emph{et al.}~\cite{aharonov2016quantum}.

\subsection{Adjudicating the operational assumptions underlying  the pigeonhole paradox}

As in the $N$-box PPS scenario, we can distinguish two aspects of the pigeonhole scenario: a purely \emph{operational} one, detached from any interpretation and simply characterized by the statistical predictions entering the ABL rule, together with additional empirical constraints, and an \emph{interpretational} one, where the paradoxical conclusion is drawn only once further classical assumptions are imposed on top of those statistics. Isolating exactly which operational assumptions are key for a paradoxical reading is, again, essential. 

At least two such assumptions are required, mirroring those identified by Maroney~\cite{Maroney1,MARONEY2} for the three box scenario. First, we must require that exactly three particles (never fewer) are always present. The classical pigeonhole bound of Eq.~\eqref{eq:pigeon-ineq} is a fact about \emph{three} objects distributed over \emph{two} boxes. If one allows for the possibility that two particles be always present, and one is only made to believe that there are three particles, there is no paradox associated with the fact that we never ``see''---i.e., infer using the ABL rule---two particles in the same box. This requirement mirrors the constraint imposed on Eq.~\eqref{eq:no_cheating_condition} in the three-box paradox, which is in turn associated with avoiding the type of ``cheating'' classical toy model introduced by Maroney~\cite{MARONEY2}. Operationally, this is expressed by requiring a specification of the scenario prior to the post-selection, allowing one to check that sequential intermediate measurements never find one of the three particles missing. 

The pigeonhole scenario therefore requires a different operationalization than the $N$-box scenario, since here the classical intuition to be tested is not about a single particle's location, but about how many of \emph{three} distinguishable particles can be found sharing a box. Notwithstanding this structural difference, it is still possible to formulate a directly testable operational signature that mirrors Eq.~\eqref{eq:no_cheating_condition}. For any pre-selection, 
\begin{align}
p({\rm diff}_{12},{\rm diff}_{13},{\rm diff}_{23}&\mid\rho_I,\mathcal{M}_{12},\mathcal{M}_{13},\mathcal{M}_{23})= 0 \label{eq:no_cheating_pigeonhole},
\end{align}
and likewise for any other order in which the three pair-checks are performed. Operationally, Eq.~\eqref{eq:no_cheating_pigeonhole} means that in any single run of the experiment (without the post-selection), it is impossible to find all three pairs of particles in different boxes. For the specific case shown, this means that sequentially finding particles 1 and 2 in different boxes, and particles 2 and 3 in different boxes already guarantees that particles 1 and 3 will be found in the same box. The third check can never independently report that the final pair of particles are in different boxes. 

Second, we must still abide to some notion of operational non-disturbance of the post-selection probability. As in the $N$-box case, any intermediate measurement $\mathcal{M}_{ij}$ could in principle disturb the system in a way that alters the probability of successful post-selection, and such disturbance---if detectable in the final statistics---may allow for a classical toy model to reproduce the statistics of the quantum pigeonhole principle while evading paradoxical conclusions (similarly to the discussion present in Refs.~\cite{LeiferSpekkensontology,ravon2007three}). We therefore require once more a form of operational non-disturbance condition, that for each pair $\{i,j\}$,
\begin{align}
&p({\rm same}_{ij},{\rm pass}) + p({\rm diff}_{ij},{\rm pass}) = p({\rm pass}\mid\rho_I,\mathcal{M}_F), \label{eq:pigeon-nondist}
\end{align}
i.e.\ the probability that post-selection succeeds does not depend on whether (or which) pair-check was performed at the intermediate step. Above, we have written $p({\rm same}_{ij},{\rm pass}) \equiv p({\rm same}_{ij},{\rm pass} \mid \rho_I,\mathcal{M}_{ij},\mathcal{M}_F)$, and also for $p({\rm diff}_{ij},{\rm pass})$,  for brevity. Exactly as in Sec.~\ref{sec:scenario}, this condition is automatically satisfied whenever either $\rho_I$ or $\rho_F$ is incoherent with respect to the relevant basis (here, the eigenbasis of $\Pi^{\rm same}_{ij}$ for each pair), so imposing it costs nothing on the classical side of the argument.

As before, we organize these in a definition, for future referencing.

\begin{definition}[Pigeonhole PPS scenario and paradoxes]\label{defPigeon}
A \emph{pigeonhole PPS scenario} is defined by a pre-selected state $\rho_I$, a post-selected state $\rho_F$, both on $\mathcal{H}=(\mathbbm{C}^2)^{\otimes3}$, and the product basis $C=\{\ket{c_1c_2c_3}\}_{c_i\in\{L,R\}}$ (with $\Pi_c=\ketbra{c}{c}$) for the intermediate measurements. Write 
\begin{equation}\label{eq:labelset}
P_3 := \big\{\{i,j\} : i,j\in\{1,2,3\},\ i\neq j\big\}
\end{equation}
for the set of unordered pairs of particle indices. For each pair $\{i,j\}\in P_3$, the intermediate measurement $\mathcal{M}_{ij}=\{\Pi^{\rm same}_{ij},\Pi^{\rm diff}_{ij}\}$ is performed, where
\begin{equation}\label{eq:pigeonhole_intermediate}
\Pi^{\rm same}_{ij}=\ketbra{LL}{LL}_{ij}\otimes\mathbb{1}_2+\ketbra{RR}{RR}_{ij}\otimes\mathbb{1}_2,
\end{equation}
and
\begin{equation}\label{eq:pigeonhole_diff}
\Pi^{\rm diff}_{ij} := \mathbb{1}_8-\Pi^{\rm same}_{ij}.
\end{equation}
Furthermore, the statistics must satisfy Eq.~\eqref{eq:no_cheating_pigeonhole}, the operational non-disturbance condition~\eqref{eq:pigeon-nondist} for each pair $\{i,j\}$, and that $p({\rm pass}\mid\rho_I,\mathcal{M}_F)\neq0$. We say that the ABL rule leads to a \emph{pigeonhole paradox} in this scenario whenever Ineq.~\eqref{eq:pigeon-ineq} is violated. Additionally, we say that it leads to a \emph{logical} pigeonhole paradox whenever $p({\rm same}_{ij},{\rm pass})=0$ for all three pairs $\{i,j\} \in P_3$ simultaneously.
\end{definition}

\section{Quantum pigeonhole principle and coherence}

As in Sec.~\ref{sec:results}, incoherence of either the pre-selected $\rho_I$ or the post-selected $\rho_F$ states with respect to the product basis $C=\{\ket{c_1c_2c_3}\}_{c_i\in\{L,R\}}$ forces the non-disturbance condition~\eqref{eq:pigeon-nondist} automatically, for any pair $\{i,j\}$. To see this, if we write $$\rho_I=\sum_{\vert c_1c_2c_3\rangle \in C}p_{(c_1,c_2,c_3)}\ketbra{c_1c_2c_3}{c_1c_2c_3},$$ noting that $\Pi^{\rm same}_{ij}$ is a sum of basis projectors of $C$---and therefore, is itself diagonal in $C$---one has $$\Pi^{\rm same}_{ij}\rho_I\Pi^{\rm same}_{ij}+\Pi^{\rm diff}_{ij}\rho_I\Pi^{\rm diff}_{ij}=\rho_I.$$ Multiplying by $\rho_F$ and taking the trace on both sides recovers Eq.~\eqref{eq:pigeon-nondist} for any $\rho_F$. 

Given the operational non-disturbance condition~\eqref{eq:pigeon-nondist}, the same anticommutator identity used in Eqs.~\eqref{eq:quant_non_dist_alt}-\eqref{eq:quant_non_dist_alt4} applies verbatim to $\Pi^{\rm same}_{ij}$---the argument only used that $\Pi^{\rm same}_{ij}$ is an orthogonal projector, not that it has rank one---giving
\begin{align}
p({\rm same}_{ij},{\rm pass}) &=\mathrm{Tr}[\rho_F \Pi_{ij}^{\rm same} \rho_I \Pi_{ij}^{\rm same}]\\
&= \mathrm{Re}\big(\mathrm{Tr}[\Pi^{\rm same}_{ij}\rho_F\rho_I]\big). \label{eq:pigeon-antic}
\end{align}
Writing 
\begin{equation}\label{eq:theta_operator}
\Theta:=\Pi^{\rm same}_{12}+\Pi^{\rm same}_{13}+\Pi^{\rm same}_{23},
\end{equation}
and summing Eq.~\eqref{eq:pigeon-antic} over the three pairs $\{i,j\} \in P_3$ (see Eq.~\eqref{eq:labelset}) gives
\begin{align}
\sum_{\{i,j\} \in P_3}p({\rm same}_{ij},{\rm pass}) &= \mathrm{Re}\big(\mathrm{Tr}[\Theta\rho_F\rho_I]\big) \nonumber \\&= \mathrm{Tr}[\rho_F\rho_I]\,\mathrm{Re}\big(\langle\Theta\rangle_w\big),\label{eq:pigeon-sum-antic}
\end{align}
using the definition of the weak value, Eq.~\eqref{eq:wv}, for the observable $\Theta$.

\begin{theorem}[A coherence-free subtheory cannot yield a pigeonhole paradox]\label{thm:pigeon-nonlogical}
For the pigeonhole scenario of Def.~\ref{defPigeon}, if either $\rho_I$ or $\rho_F$ is incoherent with respect to the product basis $C$, then Ineq.~\eqref{eq:pigeon-ineq} is satisfied.
\end{theorem}
\begin{proof}
We first show that $\Theta$ has spectrum $\{1,3\} \subset \mathbbm{N}$. Since each $\Pi^{\rm same}_{ij}$ is diagonal in the product basis $C$, so is $\Theta$, and its eigenvalues are obtained by evaluating $\Theta$ on each basis string $\vert c\rangle = \vert c_1c_2c_3 \rangle$. The eigenvalue of $\Theta$ at $c$ equals the number of pairs $\{i,j\}$ for which $c_i = c_j$. As noted earlier, this count is $3$ for the two uniform strings ($LLL,RRR$, where all three pairs coincide) and $1$ for each of the remaining six, non-uniform strings (where exactly one pair coincides). Hence, the diagonalization of $\Theta$ relative to $C$ is given by  
\begin{equation}
\Theta=1\cdot\Pi_{\lambda=1}+3\cdot\Pi_{\lambda=3}, 
\end{equation}
with $\Pi_{\lambda=1}$ the rank-six projector onto the non-uniform strings and $\Pi_{\lambda=3}=\Pi_{LLL}+\Pi_{RRR}$. 

As in the proof of~\cref{thm:nonlogical}, using the results of Ref.~\cite{WG23}, if either $\rho_F$ or $\rho_I$ is incoherent with respect to $C$, then the weak values $\langle\Pi_c\rangle_w$ of the rank-one basis projectors $\Pi_c=\ketbra{c}{c}$, $c\in C$, form a valid probability distribution, and hence so does any coarse-graining of them, including $\langle\Pi^{\rm same}_{ij}\rangle_w$ for each pair and $\langle\Theta\rangle_w=\sum_{\{i,j\}}\langle\Pi^{\rm same}_{ij}\rangle_w$. Consequently $\langle\Theta\rangle_w$ lies within the eigenvalue range of $\Theta$,
\begin{equation}
1 \leq \langle\Theta\rangle_w \leq 3,
\end{equation}
which, used in Eq.~\eqref{eq:pigeon-sum-antic}, gives
\begin{equation}
\sum_{\{i,j\} \in P_3}p({\rm same}_{ij},{\rm pass}) \geq \mathrm{Tr}[\rho_F\rho_I] = p({\rm pass}\mid\rho_I,\mathcal{M}_F).
\end{equation}
Using the non-disturbance condition~\eqref{eq:pigeon-nondist} for each pair together with the ABL rule~\eqref{eq:pigeon-ABL}, exactly as in the proof of \cref{thm:nonlogical},
\begin{align}
&\sum_{\{i,j\}}p({\rm same}_{ij}\mid\rho_I,\rho_F,\mathcal{M}_{ij}) \stackrel{\eqref{eq:pigeon-ABL}}{=} \nonumber\\
&\sum_{\{i,j\}}\frac{p({\rm same}_{ij},{\rm pass})}{p({\rm same}_{ij},{\rm pass})+p({\rm diff}_{ij},{\rm pass})} \stackrel{\eqref{eq:pigeon-nondist}}{=} \nonumber\\
&\sum_{\{i,j\}}\frac{p({\rm same}_{ij},{\rm pass})}{p({\rm pass}\mid\rho_I,\mathcal{M}_F)} \geq \frac{p({\rm pass}\mid\rho_I,\mathcal{M}_F)}{p({\rm pass}\mid\rho_I,\mathcal{M}_F)} = 1.
\end{align}
The above shows that for a paradox to happen we must have that $\mathrm{Re}\left[\langle \Theta \rangle_w \right] < 1$, which is an anomalous weak value for that operator. Because any such value requires coherence~\cite{WG23} of both $\rho_I$ and $\rho_F$, this completes the proof.
\end{proof}

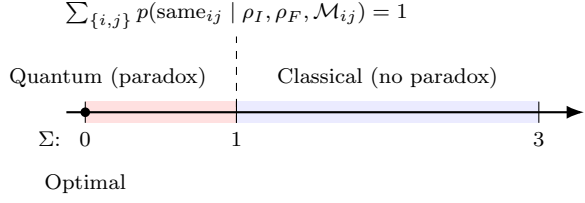
\begin{figure}[t]
\centering
\begin{tikzpicture}[>=Latex, font=\small]
\fill[red!12] (0,-0.14) rectangle (2,0.14);
\fill[blue!8] (2,-0.14) rectangle (6,0.14);
\node[font=\footnotesize] at (0.3,0.46) {Quantum (paradox)};
\node[font=\footnotesize] at (4,0.46) {Classical (no paradox)};
\draw[dashed] (2,-0.14)--(2,1);
\node[above, font=\footnotesize, align=center] at (2,0.98)
  {$\sum_{\{i,j\}}p({\rm same}_{ij}\mid\rho_I,\rho_F,\mathcal{M}_{ij})=1$};
\draw[->,thick] (-0.25,0)--(6.6,0);
\foreach \x/\s in {0/0, 2/1, 6/3}{
  \draw (\x,0.14)--(\x,-0.14);
  \node[font=\footnotesize] at (\x,-0.35) {$\s$};
}
\node[left, font=\footnotesize] at (-0.2,-0.35) {$\Sigma$:};
\fill (0,0) circle (1.8pt);
\node[font=\footnotesize] at (0,-0.94) {Optimal};
\end{tikzpicture}
\caption{\textbf{The pigeonhole paradox as a function of $\mathrm{Re}\big(\langle\Theta\rangle_w\big)$.} Under the operational non-disturbance condition~\eqref{eq:pigeon-nondist}, the pigeonhole statistics obey $\Sigma\equiv\sum_{\{i,j\}}p({\rm same}_{ij}\mid\rho_I,\rho_F,\mathcal{M}_{ij})=\mathrm{Re}\big(\langle\Theta\rangle_w\big)$ (\cref{coro:pigeon-anomalous}), where $\Theta=\sum_{\{i,j\}}\Pi^{\rm same}_{ij}$ has spectrum $\{1,3\}$. The classicality bound $\Sigma\geq1$, Ineq.~\eqref{eq:pigeon-ineq}, is violated in the shaded ``paradox'' region, where $\Theta$ has an anomalous weak value $\mathrm{Re}\big(\langle\Theta\rangle_w\big)<1$, falling below its minimum eigenvalue. The logical pigeonhole paradox of \cref{ex:pigeonhole} saturates the algebraic minimum $\Sigma=0$, i.e.\ $\langle\Theta\rangle_w=0$.}
\label{fig:pigeonbound}
\end{figure}

These results allow us to connect once more anomalous weak values of a certain observable (in this case, the operator $\Theta$) to the presence of a non-logical paradox in the scenario of the quantum pigeonhole principle. 

\begin{corollary}[Anomalous weak values of $\Theta$ quantify the non-logical pigeonhole paradox]
\label{coro:pigeon-anomalous}
For the pigeonhole scenario of Def.~\ref{defPigeon}, the sum of ABL statistics satisfies
\begin{equation}\label{eq:pigeon-sumweakvalue}
\sum_{\{i,j\}\in P_3}p({\rm same}_{ij}\mid\rho_I,\rho_F,\mathcal{M}_{ij}) = \mathrm{Re}\big(\langle\Theta\rangle_w\big).
\end{equation}
Consequently, a pigeonhole paradox occurs if and only if the weak value of $\Theta$ falls below its minimum eigenvalue,
\begin{equation}
\mathrm{Re}\big(\langle\Theta\rangle_w\big) < 1,
\end{equation}
that is, if and only if $\langle\Theta\rangle_w$ is anomalous.
\end{corollary}

Interestingly, in this case, even positive anomalous weak values $0\leq \mathrm{Re}\big(\langle\Theta\rangle_w\big)<1$ are capable of yielding a paradox. We can now revisit the original pigeonhole state considered in Sec.~\ref{sec:pigeon} through the lens of \cref{coro:pigeon-anomalous}.

\begin{exa}[Weak value of $\Theta$ and the original pigeonhole paradox]
\label{ex:pigeonhole}
For $\ket{\psi_I}=\ket{+}^{\otimes3}$ and $\ket{\psi_F}=\ket{+_i}^{\otimes3}$, with $\ket{+}=(\ket{L}+\ket{R})/\sqrt{2}$, $\ket{+_i}=(\ket{L}+i\ket{R})/\sqrt{2}$, one has a post-selection with probability $\mathrm{Tr}(\rho_I\rho_F)=\vert\braket{\psi_F}{\psi_I}\vert^2=\sfrac{1}{8}$. Individually, each pairwise weak value vanishes~\cite{aharonov2016quantum}, so that, summing over the three pairs,
\begin{equation}
\langle \Theta\rangle_w = \sum_{\{i,j\}}\left\langle \Pi^{\rm same}_{ij}\right\rangle_w = 0.
\end{equation}
By \cref{coro:pigeon-anomalous}, $\mathrm{Re}(\langle\Theta\rangle_w)=0<1$, so $\langle\Theta\rangle_w$ is anomalous and thus yielding the optimal 
\begin{equation}
\sum_{\{i,j\}}p({\rm same}_{ij}\mid\rho_I,\rho_F,\mathcal{M}_{ij}) = \mathrm{Re}(\langle\Theta\rangle_w) = 0,
\end{equation}
paradox in the quantum pigeonhole scenario. 
\end{exa}

This quantifiable characteristic of a non-logical form of the quantum pigeonhole paradox provides an extension to mixed states of the results in Refs.~\cite{aharonov2016quantum,waegell2017confined,waegell2017contextuality,waegell2023negativemass,correa2021apparent}. Similarly as before, this provides a robustness criterion for observing a paradox even if the pre- and post-selection states are not identically equal to $\vert \psi_I \rangle$ or $\vert \psi_F \rangle$. 

\begin{exa}[A family of non-logical pigeonhole paradoxes]
\label{ex:pigeonhole-family}
We now construct an example of a paradox with mixed pre- and post-selected states. Let
\begin{align}
\rho_I &= p_I\,\ket{\psi_I}\bra{\psi_I} + (1-p_I)\,\ket{LLL}\bra{LLL},\\
\rho_F &= p_F\,\ket{\psi_F}\bra{\psi_F} + (1-p_F)\,\ket{LLL}\bra{LLL},
\end{align}
where $\ket{\psi_I},\ket{\psi_F}$ are taken to be the same as in \cref{ex:pigeonhole} and $0\leq p_I,p_F<1$. A direct calculation, using $\braket{LLL}{\psi_I}=\braket{LLL}{\psi_F}=\sfrac{1}{\sqrt{8}}$ and $\braket{\psi_F}{\Pi^{\rm same}_{ij}\psi_I}=0$ for every pair (\cref{ex:pigeonhole}), gives, for each pair $(i,j)$,
\begin{align}
p({\rm same}_{ij},{\rm pass}) &= 1 - \tfrac{7}{8}(p_I+p_F) + \tfrac{3}{4}p_Ip_F,\\
p({\rm diff}_{ij},{\rm pass}) &= \tfrac{1}{8}p_Ip_F,\\
p({\rm pass}\mid\rho_I,\mathcal{M}_F) &= 1 - \tfrac{7}{8}(p_I+p_F) + \tfrac{7}{8}p_Ip_F,
\end{align}
independent of which pair is chosen. Furthermore, this statistics satisfies the operational non-disturbance condition~\eqref{eq:pigeon-nondist} for all values of $p_I$ and $p_F$. One obtains the ABL probabilities as 
\begin{equation}
p({\rm same}_{ij}\mid\rho_I,\rho_F,\mathcal{M}_{ij}) = \frac{8-7(p_I+p_F)+6p_Ip_F}{8-7(p_I+p_F)+7p_Ip_F},
\end{equation}
the same for each of the three pairs. Consequently, via Eq.~\eqref{eq:pigeon-sumweakvalue} we find that the real part of the weak value $\mathrm{Re}[\langle\Theta\rangle_w]= \sum_{\{i,j\} \in P_3} p({\rm same}_{ij}\mid\rho_I,\rho_F,\mathcal{M}_{ij})$ is given by
\begin{equation}
\mathrm{Re}[\langle\Theta\rangle_w] = \frac{3\big[8-7(p_I+p_F)+6p_Ip_F\big]}{8-7(p_I+p_F)+7p_Ip_F}.
\end{equation}

On the other hand, computing $\langle\Theta\rangle_w=\mathrm{Tr}[\rho_F\Theta\rho_I]/\mathrm{Tr}[\rho_F\rho_I]$ directly from Eq.~\eqref{eq:wv}, yields independently the ABL statistics above, validating Eq.~\eqref{eq:pigeon-sumweakvalue}. Expanding $\rho_I,\rho_F$ into their four cross-terms, using $\langle\psi_F|\Theta|\psi_I\rangle=0$ (\cref{ex:pigeonhole}) and $\Theta\ket{LLL}=3\ket{LLL}$, gives $$\mathrm{Tr}[\rho_F\Theta\rho_I]=3\big[1-\tfrac78(p_I+p_F)+\tfrac34p_Ip_F\big]$$ and $$\mathrm{Tr}[\rho_F\rho_I]=1-\tfrac78(p_I+p_F)+\tfrac78p_Ip_F,$$ recovering the same expression for $\langle\Theta\rangle_w$ exactly.

The paradox appears when $\mathrm{Re}(\langle\Theta\rangle_w)<1$. In this case, the region of $p_I$ and $p_F$ satisfying this criterion is given by
\begin{equation}
14(p_I+p_F) - 11\,p_Ip_F > 16.
\end{equation}
For all such values one has a non-logical instance of a pigeonhole paradox. At $p_I=p_F=1$ this recovers the logical pigeonhole paradox of \cref{ex:pigeonhole} ($\langle\Theta\rangle_w=0$), while along the diagonal $p_I=p_F=p$ the paradox persists down to $p=\tfrac{1}{11}(14-2\sqrt5)\approx0.866$, below which the statistics become classical.
\end{exa}

It is also worth mentioning that even a real, nonnegative anomalous weak value of $\Theta$ still implies coherence, as captured by Bargmann invariants. This is because, in such a case, at least one of the Bargmann invariants constructed out of some projector $\Pi_{c_1c_2c_3}$---with $c_i \in \{L,R\}$---together with $\rho_I$ and $\rho_F$ must have a negative real part or nonzero imaginary part~\cite{WG23}. For example, in the original pigeonhole paradox we find, writing $n_R(c)$ for the number of $R$'s in the string $c=c_1c_2c_3$,
\begin{equation}
\langle c \vert \psi_I \rangle  = \frac{1}{2\sqrt{2}}, \qquad \braket{\psi_F}{c} = \frac{(-i)^{n_R(c)}}{2\sqrt{2}},
\end{equation}
and $\langle \psi_I \vert \psi_F \rangle =(-1+i)/4$, so that the third-order Bargmann invariant associated with each basis projector $\Pi_c=\vert c \rangle \langle c \vert$ is
\begin{align}
\Delta(\rho_F,\Pi_c,\rho_I) &= \braket{\psi_F}{c}\braket{c}{\psi_I}\braket{\psi_I}{\psi_F} \nonumber \\
&= \frac{(-i)^{n_R(c)}(-1+i)}{32}.
\end{align}
Evaluating this for each of the four possible values of $n_R(c)$ gives
\begin{equation}
\mathrm{Re}\big(\Delta(\rho_F,\Pi_c,\rho_I)\big) =
\begin{cases}
-\tfrac{1}{32}, & n_R(c)=0 \text{ or } 3,\\[2pt]
+\tfrac{1}{32}, & n_R(c)=1 \text{ or } 2.
\end{cases}
\end{equation}
Thus, exactly the two uniform strings $c=LLL$ and $c=RRR$---the strings with $\Theta$-eigenvalue $3$, i.e.\ the two configurations for which \emph{all three} pairs coincide---carry a negative Bargmann invariant, while all six non-uniform strings (each contributing $\Theta$-eigenvalue $1$) carry a positive one. 

\section{Conclusions}\label{sec:conclusions}

We have shown that quantum coherence is a necessary resource for logical and non-logical $N$-box and quantum pigeonhole PPS paradoxes, and that its role is not merely qualitative but exactly quantifiable. For the $N$-box scenario, \cref{thm:nonlogical} shows that a coherence-free subtheory---one in which either the pre- or the post-selected state is incoherent with respect to the box basis---cannot violate the classicality bound~\eqref{eq:probNbox}, while \cref{coro:anomalous} identifies the precise quantity governing any such violation: the weak value of the excluded box, $\langle\Pi_{N+1}\rangle_w$. A paradox can occur if and only if this weak value is anomalous.

As we elaborate in Sec.~\ref{sec:LG}, our results align with Maroney's identification of the three-box paradox as a violation of a Leggett--Garg inequality~\cite{MARONEY2}. Building on Schmid's reformulation of macrorealism as strict classicality~\cite{schmid2024reviewreformulation}, we note that, within quantum theory, a strictly classical fragment corresponds precisely to an incoherent subtheory relative to a fixed basis~\cite{selby2023accessible,schmid2024reviewreformulation}. Under this correspondence,~\cref{thm:nonlogical} is consistent with Maroney's conclusions, since coherence is necessary for violating a Leggett--Garg inequality.

We established the analogous pair of results for the quantum pigeonhole principle in \cref{thm:pigeon-nonlogical,coro:pigeon-anomalous}, where the governing quantity is instead the weak value of the single operator $\Theta=\sum_{\{i,j\}}\Pi^{\rm same}_{ij}$, and a paradox occurs precisely when $\langle\Theta\rangle_w$ falls below its minimum eigenvalue. In both cases, the underlying mechanism is the same: non-disturbance, together with the anticommutator identity of Eq.~\eqref{eq:quant_non_dist_alt}, converts the relevant ABL statistics into the real part of a weak value, whose anomaly is required for the paradox and indicates the relevance of coherence.

Since logical PPS paradoxes are already known to be proofs of contextuality~\cite{LeiferspekkensPPScontextuality,PuseyLeifer2015} in its standard form~\cite{kochen1967problem,budroni2022kochenspecker}, a natural next step is to determine whether the non-logical $N$-box and pigeonhole paradoxes studied here are themselves proofs of generalized contextuality~\cite{spekkens2005contextuality}. In particular, it remains an open question whether classical (though \emph{not} strictly classical~\cite{schmid2024reviewreformulation}) toy models can reproduce the predictions of PPS scenarios such as the ones considered here. A worthwhile direction for future work is to clarify whether such models can reproduce the operational predictions of these PPS scenarios (and others, such as the quantum Cheshire cat~\cite{aharonov2013cheshire}) in some regimes, or, alternatively, whether this is impossible because the inequalities considered here, together with the operational descriptions characterizing the associated PPS scenarios, can themselves be shown to be generalized noncontextuality inequalities.

\begin{acknowledgments}
EFG acknowledges support from FCT – Fundaç\~{a}o para a Ciência e a Tecnologia (Portugal) via project CEECINST/00062/2018 and from the National Council for
Scientific and Technological Development – CNPq (Brazil) under grant 308292/2025-1. RW acknowledges support from the Alexander von Humboldt Foundation. SK acknowledges funding from the Digital Horizon Europe project \href{https://cordis.europa.eu/project/id/101070558}{FoQaCiA} (\textit{Foundations of Quantum Computational Advantage}), GA no.101070558, funded by the European Union and NSERC (Canada). 
\end{acknowledgments}

\bibliography{bibliography}

@PREAMBLE{
 "\providecommand{\noopsort}[1]{}" 
 # "\providecommand{\singleletter}[1]{#1}%" 
}

@article{budroni2022kochenspecker, title={{Kochen-Specker contextuality}}, volume={94}, ISSN={1539-0756}, url={http://dx.doi.org/10.1103/RevModPhys.94.045007}, DOI={10.1103/revmodphys.94.045007}, number={4}, journal={Rev. Mod. Phys.}, publisher={American Physical Society (APS)}, author={Budroni, Costantino and Cabello, Adán and Gühne, Otfried and Kleinmann, Matthias and Larsson, Jan-Ake}, year={2022}, month=Dec,pages={045007} }

@article{baumgratz2014quantifying,
  title = {{Quantifying Coherence}},
  author = {Baumgratz, T. and Cramer, M. and Plenio, M. B.},
  journal = {Phys. Rev. Lett.},
  volume = {113},
  issue = {14},
  pages = {140401},
  numpages = {5},
  year = {2014},
  month = {Sep},
  publisher = {American Physical Society},
  doi = {10.1103/PhysRevLett.113.140401},
  url = {https://link.aps.org/doi/10.1103/PhysRevLett.113.140401}
}

@article{emary2013leggett, title={{Leggett--Garg inequalities}}, volume={77}, ISSN={1361-6633}, url={http://dx.doi.org/10.1088/0034-4885/77/1/016001}, DOI={10.1088/0034-4885/77/1/016001}, number={1}, journal={Rep. Prog. Phys.}, publisher={IOP Publishing}, author={Emary, Clive and Lambert, Neill and Nori, Franco}, year={2013}, month=Dec, pages={016001} }

@article{schmid2024reviewreformulation,
  doi = {10.22331/q-2024-01-03-1217},
  url = {https://doi.org/10.22331/q-2024-01-03-1217},
  title = {A review and reformulation of macroscopic realism: resolving its deficiencies using the framework of generalized probabilistic theories},
  author = {Schmid, David},
  journal = {{Quantum}},
  issn = {2521-327X},
  publisher = {{Verein zur F{\"{o}}rderung des Open Access Publizierens in den Quantenwissenschaften}},
  volume = {8},
  pages = {1217},
  month = jan,
  year = {2024}
}

@article{vitagliano2023leggett,
  title = {{Leggett--Garg macrorealism and temporal correlations}},
  author = {Vitagliano, Giuseppe and Budroni, Costantino},
  journal = {Phys. Rev. A},
  volume = {107},
  issue = {4},
  pages = {040101},
  numpages = {22},
  year = {2023},
  month = {Apr},
  publisher = {American Physical Society},
  doi = {10.1103/PhysRevA.107.040101},
  url = {https://link.aps.org/doi/10.1103/PhysRevA.107.040101}
}

@article{leggett1985quantum,
  title = {Quantum mechanics versus macroscopic realism: Is the flux there when nobody looks?},
  author = {Leggett, A. J. and Garg, Anupam},
  journal = {Phys. Rev. Lett.},
  volume = {54},
  issue = {9},
  pages = {857--860},
  numpages = {0},
  year = {1985},
  month = {Mar},
  publisher = {American Physical Society},
  doi = {10.1103/PhysRevLett.54.857},
  url = {https://link.aps.org/doi/10.1103/PhysRevLett.54.857}
}

@article{kastner2003nature, title={{The Nature of the Controversy over Time-Symmetric Quantum Counterfactuals}}, volume={70}, ISSN={1539-767X}, url={http://dx.doi.org/10.1086/367874}, DOI={10.1086/367874}, number={1}, journal={Philos. Sci.}, publisher={Cambridge University Press (CUP)}, author={Kastner, Ruth E.}, year={2003}, month=Jan, pages={145–163} }

@article{tollaksen2007pre, title={Pre- and post-selection, weak values and contextuality}, volume={40}, ISSN={1751-8121}, url={http://dx.doi.org/10.1088/1751-8113/40/30/025}, DOI={10.1088/1751-8113/40/30/025}, number={30}, journal={ J. Phys. A: Math. Theor.}, publisher={IOP Publishing}, author={Tollaksen, Jeff}, year={2007}, month=July, pages={9033–9066} }

@incollection{aharonov2002time,
  author    = {Aharonov, Yakir and Vaidman, Lev},
  title     = {The Two-State Vector Formalism of Quantum Mechanics},
  booktitle = {Time in Quantum Mechanics},
  editor    = {Muga, J. G. and Sala Mayato, R. and Egusquiza, I. L.},
  series    = {Lecture Notes in Physics},
  volume    = {72},
  pages     = {369--412},
  publisher = {Springer},
  address   = {New York},
  year      = {2002},
  doi       = {10.1007/3-540-45846-8_13},
  url       = {https://doi.org/10.1007/3-540-45846-8_13}
}

@article{WG23,
   title={Simple proof that anomalous weak values require coherence},
   volume={108},
   ISSN={2469-9934},
   url={http://dx.doi.org/10.1103/PhysRevA.108.L040202},
   DOI={10.1103/physreva.108.l040202},
   number={4},
   journal={Phys. Rev. A},
   publisher={American Physical Society (APS)},
   author={Wagner, Rafael and Galvão, Ernesto F.},
   year={2023},
   month=oct, pages={L040202 } }

@article{albert1985curious,
  title = {{Curious New Statistical Prediction of Quantum Mechanics}},
  author = {Albert, David Z. and Aharonov, Yakir and D'Amato, Susan},
  journal = {Phys. Rev. Lett.},
  volume = {54},
  issue = {1},
  pages = {5--7},
  numpages = {0},
  year = {1985},
  month = {Jan},
  publisher = {American Physical Society},
  doi = {10.1103/PhysRevLett.54.5},
  url = {https://link.aps.org/doi/10.1103/PhysRevLett.54.5}
}

@article{waegell2017confined,
  title = {Confined contextuality in neutron interferometry: Observing the quantum pigeonhole effect},
  author = {Waegell, Mordecai and Denkmayr, Tobias and Geppert, Hermann and Ebner, David and Jenke, Tobias and Hasegawa, Yuji and Sponar, Stephan and Dressel, Justin and Tollaksen, Jeff},
  journal = {Phys. Rev. A},
  volume = {96},
  issue = {5},
  pages = {052131},
  numpages = {8},
  year = {2017},
  month = {Nov},
  publisher = {American Physical Society},
  doi = {10.1103/PhysRevA.96.052131},
  url = {https://link.aps.org/doi/10.1103/PhysRevA.96.052131}
}

@article{waegell2017contextuality, title={{Contextuality, Pigeonholes, Cheshire Cats, Mean Kings, and Weak Values}}, volume={5}, ISSN={2196-5617}, url={http://dx.doi.org/10.1007/s40509-017-0127-9}, DOI={10.1007/s40509-017-0127-9}, number={2}, journal={Quantum Stud.: Math. Found.}, publisher={Springer Science and Business Media LLC}, author={Waegell, Mordecai and Tollaksen, Jeff}, year={2017}, month=Sept, pages={325–349} }

@book{aharonov2005quantum,
  title     = {{Quantum Paradoxes: Quantum Theory for the Perplexed}},
  author    = {Aharonov, Yakir and Rohrlich, Daniel},
  publisher = {Wiley-VCH},
  address   = {Weinheim},
  year      = {2005},
  isbn      = {9783527405683},
  doi={10.1002/9783527619115}
}

@article{bub1986curious,
  title = {{Curious Properties of Quantum Ensembles Which Have Been Both Preselected and Post-Selected}},
  author = {Bub, Jeffrey and Brown, Harvey},
  journal = {Phys. Rev. Lett.},
  volume = {56},
  issue = {22},
  pages = {2337--2340},
  numpages = {0},
  year = {1986},
  month = {Jun},
  publisher = {American Physical Society},
  doi = {10.1103/PhysRevLett.56.2337},
  url = {https://link.aps.org/doi/10.1103/PhysRevLett.56.2337}
}

@article{albert1986comment,
  title = {{Comment on ``Curious Properties of Quantum Ensembles Which Have Been Both Preselected and Post-Selected'''}},
  author = {Albert, David Z. and Aharonov, Yakir and D'Amato, Susan},
  journal = {Phys. Rev. Lett.},
  volume = {56},
  issue = {22},
  pages = {2427},
  numpages = {0},
  year = {1986},
  month = {Jun},
  publisher = {American Physical Society},
  doi = {10.1103/PhysRevLett.56.2427},
  url = {https://link.aps.org/doi/10.1103/PhysRevLett.56.2427}
}

@article{leavens2006general, title={{General $N$-box problem}}, volume={359}, ISSN={0375-9601}, url={http://dx.doi.org/10.1016/j.physleta.2006.06.089}, DOI={10.1016/j.physleta.2006.06.089}, number={5}, journal={Phys. Lett. A}, publisher={Elsevier BV}, author={Leavens, C.R. and Puerto Gimenez, I. and Alonso, D. and Sala Mayato, R.}, year={2006}, month=Dec, pages={416–423} }

@article{aharonov1991complete, title={Complete description of a quantum system at a given time}, volume={24}, ISSN={1361-6447}, url={http://dx.doi.org/10.1088/0305-4470/24/10/018}, DOI={10.1088/0305-4470/24/10/018}, number={10}, journal={J. Phys. A: Math. Gen. }, publisher={IOP Publishing}, author={Aharonov, Y and Vaidman, L}, year={1991}, month=May, pages={2315–2328} }

@article{hance2023contextuality, title={{Contextuality, coherences, and quantum Cheshire cats}}, volume={25}, ISSN={1367-2630}, url={http://dx.doi.org/10.1088/1367-2630/ad0bd4}, DOI={10.1088/1367-2630/ad0bd4}, number={11}, journal={New J. Phys.}, publisher={IOP Publishing}, author={Hance, Jonte R and Ji, Ming and Hofmann, Holger F}, year={2023}, month=Nov, pages={113028} }

@article{hance2024dynamicalcheshire, title={{Is the dynamical quantum Cheshire cat detectable?}}, volume={26}, ISSN={1367-2630}, url={http://dx.doi.org/10.1088/1367-2630/ad6476}, DOI={10.1088/1367-2630/ad6476}, number={7}, journal={New J. Phys.}, publisher={IOP Publishing}, author={Hance, Jonte R and Ladyman, James and Rarity, John}, year={2024}, month=July, pages={073038} }

@article{aharonov2013cheshire, title={{Quantum Cheshire Cats}}, volume={15}, ISSN={1367-2630}, url={http://dx.doi.org/10.1088/1367-2630/15/11/113015}, DOI={10.1088/1367-2630/15/11/113015}, number={11}, journal={New J. Phys.}, publisher={IOP Publishing}, author={Aharonov, Yakir and Popescu, Sandu and Rohrlich, Daniel and Skrzypczyk, Paul}, year={2013}, month=Nov, pages={113015} }

@article{aharonov2021dynamical, title={{A dynamical quantum Cheshire Cat effect and implications for counterfactual communication}}, volume={12}, ISSN={2041-1723}, url={http://dx.doi.org/10.1038/s41467-021-24933-9}, DOI={10.1038/s41467-021-24933-9}, number={1}, journal={Nat. Commun.}, publisher={Springer Science and Business Media LLC}, author={Aharonov, Yakir and Cohen, Eliahu and Popescu, Sandu}, year={2021}, month=Aug,pages={4770} }

@article{denkmayr2014observation, title={{Observation of a quantum Cheshire Cat in a matter-wave interferometer experiment}}, volume={5}, ISSN={2041-1723}, url={http://dx.doi.org/10.1038/ncomms5492}, DOI={10.1038/ncomms5492}, number={1}, journal={Nat. Commun.}, publisher={Springer Science and Business Media LLC}, author={Denkmayr, Tobias and Geppert, Hermann and Sponar, Stephan and Lemmel, Hartmut and Matzkin, Alexandre and Tollaksen, Jeff and Hasegawa, Yuji}, year={2014}, month=July,pages={4492} }

@article{hardy1992quantum,
  title = {{Quantum mechanics, local realistic theories, and Lorentz-invariant realistic theories}},
  author = {Hardy, Lucien},
  journal = {Phys. Rev. Lett.},
  volume = {68},
  issue = {20},
  pages = {2981--2984},
  numpages = {0},
  year = {1992},
  month = {May},
  publisher = {American Physical Society},
  doi = {10.1103/PhysRevLett.68.2981},
  url = {https://link.aps.org/doi/10.1103/PhysRevLett.68.2981}
}

@article{abramsky2012logical,
  title = {{Logical Bell inequalities}},
  author = {Abramsky, Samson and Hardy, Lucien},
  journal = {Phys. Rev. A},
  volume = {85},
  issue = {6},
  pages = {062114},
  numpages = {11},
  year = {2012},
  month = {Jun},
  publisher = {American Physical Society},
  doi = {10.1103/PhysRevA.85.062114},
  url = {https://link.aps.org/doi/10.1103/PhysRevA.85.062114}
}

@article{LeiferspekkensPPScontextuality,
   title={{Pre- and Post-Selection Paradoxes and Contextuality in Quantum Mechanics}},
   volume={95},
   ISSN={1079-7114},
   url={http://dx.doi.org/10.1103/PhysRevLett.95.200405},
   DOI={10.1103/physrevlett.95.200405},
   number={20},
   journal={Phys. Rev. Lett.},
   publisher={American Physical Society (APS)},
   author={Leifer, M. S. and Spekkens, Robert W.},
   year={2005},
   month=nov, pages={200405} }

@inbook{aharonov2016weakvalues, title={{Weak Values and Quantum Nonlocality}}, url={http://dx.doi.org/10.1017/CBO9781316219393.020}, DOI={10.1017/cbo9781316219393.020}, booktitle={Quantum Nonlocality and Reality}, publisher={Cambridge University Press}, author={Aharonov, Yakir and Cohen, Eliahu}, year={2016}, month=Aug, pages={305–314} }

@article{hofman2015quantum,
  title = {Quantum paradoxes originating from the nonclassical statistics of physical properties related to each other by half-periodic transformations},
  author = {Hofmann, Holger F.},
  journal = {Phys. Rev. A},
  volume = {91},
  issue = {6},
  pages = {062123},
  numpages = {11},
  year = {2015},
  month = {Jun},
  publisher = {American Physical Society},
  doi = {10.1103/PhysRevA.91.062123},
  url = {https://link.aps.org/doi/10.1103/PhysRevA.91.062123}
}

@article{hofmann2011ontherole, title={On the role of complex phases in the quantum statistics of weak measurements}, volume={13}, ISSN={1367-2630}, url={http://dx.doi.org/10.1088/1367-2630/13/10/103009}, DOI={10.1088/1367-2630/13/10/103009}, number={10}, journal={New J. Phys.}, publisher={IOP Publishing}, author={Hofmann, Holger F}, year={2011}, month=Oct, pages={103009} }

@misc{hofmann2009resolutionquantumparadoxesweak,
      title={On the resolution of quantum paradoxes by weak measurements}, 
      author={Holger F. Hofmann},
      year={2009},
      eprint={0911.0071},
      archivePrefix={arXiv},
      primaryClass={quant-ph},
      url={https://arxiv.org/abs/0911.0071}, 
}

@article{aharonov2002revisiting, title={{Revisiting Hardy’s paradox: counterfactual statements, real measurements, entanglement and weak values}}, volume={301}, ISSN={0375-9601}, url={http://dx.doi.org/10.1016/S0375-9601(02)00986-6}, DOI={10.1016/s0375-9601(02)00986-6}, number={3-4}, journal={Phys. Lett. A}, publisher={Elsevier BV}, author={Aharonov, Yakir and Botero, Alonso and Popescu, Sandu and Reznik, Benni and Tollaksen, Jeff}, year={2002}, month=Aug, pages={130–138} }

@article{LeiferSpekkensontology,
   title={Logical Pre- and Post-Selection Paradoxes, Measurement-Disturbance and Contextuality},
   volume={44},
   ISSN={1572-9575},
   url={http://dx.doi.org/10.1007/s10773-005-8975-1},
   DOI={10.1007/s10773-005-8975-1},
   number={11},
   journal={Int. J. Theor. Phys.},
   publisher={Springer Science and Business Media LLC},
   author={Leifer, M S and Spekkens, R W},
   year={2005},
   month=nov, pages={1977–1987} }

@article{aharonov2016quantum, title={Quantum violation of the pigeonhole principle and the nature of quantum correlations}, volume={113}, ISSN={1091-6490}, url={http://dx.doi.org/10.1073/pnas.1522411112}, DOI={10.1073/pnas.1522411112}, number={3}, journal={Proc. Natl. Acad. Sci. U.S.A. }, publisher={National Academy of Sciences}, author={Aharonov, Yakir and Colombo, Fabrizio and Popescu, Sandu and Sabadini, Irene and Struppa, Daniele C. and Tollaksen, Jeff}, year={2016}, month=Jan, pages={532–535} }

@article{PuseyLeifer2015,
   title={Logical pre- and post-selection paradoxes are proofs of contextuality},
   volume={195},
   ISSN={2075-2180},
   url={http://dx.doi.org/10.4204/EPTCS.195.22},
   DOI={10.4204/eptcs.195.22},
   journal={Electronic Proceedings in Theoretical Computer Science},
   publisher={Open Publishing Association},
   author={Pusey, Matthew F. and Leifer, Matthew S.},
   year={2015},
   month=nov, pages={295–306} }

@article{
Maroney1,
author = {Richard E. George  and Lucio M. Robledo  and Owen J. E. Maroney  and Machiel S. Blok  and Hannes Bernien  and Matthew L. Markham  and Daniel J. Twitchen  and John J. L. Morton  and G. Andrew D. Briggs  and Ronald Hanson },
title = {Opening up three quantum boxes causes classically undetectable wavefunction collapse},
journal = {Proc. Natl. Acad. Sci. U.S.A. },
volume = {110},
number = {10},
pages = {3777-3781},
year = {2013},
doi = {10.1073/pnas.1208374110}}

@article{MARONEY2,
title = {Measurements, disturbances and the quantum three box paradox},
journal = {Stud. Hist. Philos. Sci. Part B: Stud. Hist. Philos. of Mod. Phys.},
volume = {58},
pages = {41-53},
year = {2017},
issn = {1355-2198},
doi = {https://doi.org/10.1016/j.shpsb.2016.12.003},
url = {https://www.sciencedirect.com/science/article/pii/S1355219815300113},
author = {O.J.E. Maroney}
}

@article{aharonov1964time,
  title = {{Time Symmetry in the Quantum Process of Measurement}},
  author = {Aharonov, Yakir and Bergmann, Peter G. and Lebowitz, Joel L.},
  journal = {Phys. Rev.},
  volume = {134},
  issue = {6B},
  pages = {B1410--B1416},
  numpages = {0},
  year = {1964},
  month = {Jun},
  publisher = {American Physical Society},
  doi = {10.1103/PhysRev.134.B1410},
  url = {https://link.aps.org/doi/10.1103/PhysRev.134.B1410}
}

@article{santos2021conditions,
  title = {Conditions for logical contextuality and nonlocality},
  author = {Santos, Leonardo and Amaral, Barbara},
  journal = {Phys. Rev. A},
  volume = {104},
  issue = {2},
  pages = {022201},
  numpages = {13},
  year = {2021},
  month = {Aug},
  publisher = {American Physical Society},
  doi = {10.1103/PhysRevA.104.022201},
  url = {https://link.aps.org/doi/10.1103/PhysRevA.104.022201}
}

@article{pusey2014anomalous,
  title={Anomalous weak values are proofs of contextuality},
  author={Pusey, Matthew F.},
  fjournal={Phys. Rev. Lett.},
  journal={Phys. Rev. Lett.},
  volume={113},
  number={20},
  pages={200401},
  year={2014},
  doi={10.1103/PhysRevLett.113.200401},
  publisher={APS}
}

@article{kunjwal2019anomalous,
  title={Anomalous weak values and contextuality: robustness, tightness, and imaginary parts},
  author={Kunjwal, Ravi and Lostaglio, Matteo and Pusey, Matthew F.},
  fjournal={Phys. Rev. A},
  journal={Phys. Rev. A},
  volume={100},
  number={4},
  pages={042116},
  year={2019},
  doi={10.1103/PhysRevA.100.042116},
  publisher={APS}
}

@article{bell1966problem,
  title={On the problem of hidden variables in quantum mechanics},
  author={Bell, John S},
  fjournal={Rev. Mod. Phys.},
  journal={Rev. Mod. Phys.},
  volume={38},
  number={3},
  pages={447},
  year={1966},
  doi={10.1103/RevModPhys.38.447},
  publisher={APS}
}

@article{kochen1967problem,
  title={The problem of hidden variables in quantum mechanics},
  author={Kochen, Simon and Specker, Ernst P},
  fjournal={Journal of Mathematics and Mechanics},
  journal={J. Math. Mech.},
  volume={17},
  number={1},
  pages={59},
  year={1967},
  url={https://www.jstor.org/stable/24902153},
  publisher={JSTOR}
}

@article{spekkens2005contextuality,
  title={Contextuality for preparations, transformations, and unsharp measurements},
  author={Spekkens, Robert W},
  fjournal={Phys. Rev. A},
  journal={Phys. Rev. A},
  volume={71},
  number={5},
  pages={052108},
  year={2005},
  doi={10.1103/PhysRevA.71.052108},
  publisher={APS}
}

@article{aharonov1988result,
  title={How the result of a measurement of a component of the spin of a spin-1/2 particle can turn out to be 100},
  author={Aharonov, Yakir and Albert, David Z. and Vaidman, Lev},
  fjournal={Phys. Rev. Lett.},
  journal={Phys. Rev. Lett.},
  volume={60},
  number={14},
  pages={1351},
  year={1988},
  doi={10.1103/PhysRevLett.60.1351},
  publisher={APS}
}

@article{dressel2014colloquium,
  title={Colloquium: {U}nderstanding quantum weak values: {B}asics and applications},
  author={Dressel, Justin and Malik, Mehul and Miatto, Filippo M. and Jordan, Andrew N. and Boyd, Robert W.},
  fjournal={Rev. Mod. Phys.},
  journal={Rev. Mod. Phys.},
  volume={86},
  number={1},
  pages={307},
  year={2014},
  doi={10.1103/RevModPhys.86.307},
  publisher={APS}
}

@article{dressel2015weak,
  title={Weak values as interference phenomena},
  author={Dressel, Justin},
  fjournal={Phys. Rev. A},
  journal={Phys. Rev. A},
  volume={91},
  number={3},
  pages={032116},
  year={2015},
  doi={10.1103/PhysRevA.91.032116},
  publisher={APS}
}

@article{streltsov2017colloquium,
  title={Colloquium: Quantum coherence as a resource},
  author={Streltsov, Alexander and Adesso, Gerardo and Plenio, Martin B.},
  fjournal={Rev. Mod. Phys.},
  journal={Rev. Mod. Phys.},
  volume={89},
  number={4},
  pages={041003},
  year={2017},
  doi={10.1103/RevModPhys.89.041003},
  publisher={APS}
}

@article{designolle2021set,
  title = {Set Coherence: {B}asis-Independent Quantification of Quantum Coherence},
  author = {Designolle, S\'ebastien and Uola, Roope and Luoma, Kimmo and Brunner, Nicolas},
  journal = {Phys. Rev. Lett.},
  volume = {126},
  issue = {22},
  pages = {220404},
  numpages = {6},
  year = {2021},
  month = {Jun},
  publisher = {American Physical Society},
  doi = {10.1103/PhysRevLett.126.220404},
  url = {https://link.aps.org/doi/10.1103/PhysRevLett.126.220404}
}

@article{bargmann1964note,
  title={Note on {W}igner's theorem on symmetry operations},
  author={Bargmann, Valentine},
  fjournal={Journal of Mathematical Physics},
  journal={J. Math. Phys.},
  volume={5},
  number={7},
  pages={862},
  year={1964},
  doi={10.1063/1.1704188},
  publisher={American Institute of Physics}
}

@article{giordani2021witnesses,
  title = {Witnesses of coherence and dimension from multiphoton indistinguishability tests},
  author = {Giordani, Taira and Esposito, Chiara and Hoch, Francesco and Carvacho, Gonzalo and Brod, Daniel J. and Galv\~ao, Ernesto F. and Spagnolo, Nicol\`o and Sciarrino, Fabio},
  fjournal = {Physical Review Research},
  journal = {Phys. Rev. Res.},
  volume = {3},
  issue = {2},
  pages = {023031},
  numpages = {10},
  year = {2021},
  month = {Apr},
  publisher = {American Physical Society},
  doi = {10.1103/PhysRevResearch.3.023031},
  url = {https://link.aps.org/doi/10.1103/PhysRevResearch.3.023031}
}

@article{galvao2020quantum,
  title = {Quantum and classical bounds for two-state overlaps},
  author = {Galv\~ao, Ernesto F. and Brod, Daniel J.},
  journal = {Phys. Rev. A},
  volume = {101},
  issue = {6},
  pages = {062110},
  numpages = {12},
  year = {2020},
  month = {Jun},
  publisher = {American Physical Society},
  doi = {10.1103/PhysRevA.101.062110},
  url = {https://link.aps.org/doi/10.1103/PhysRevA.101.062110}
}

@article{oszmaniec2024measuring, title={Measuring relational information between quantum states, and applications}, volume={26}, ISSN={1367-2630}, url={http://dx.doi.org/10.1088/1367-2630/ad1a27}, DOI={10.1088/1367-2630/ad1a27}, number={1}, journal={New J. Phys.}, publisher={IOP Publishing}, author={Oszmaniec, Michał and Brod, Daniel J and Galvão, Ernesto F}, year={2024}, month=Jan, pages={013053} }

@article{aharonov1990properties,
  title = {Properties of a quantum system during the time interval between two measurements},
  author = {Aharonov, Yakir and Vaidman, Lev},
  journal = {Phys. Rev. A},
  volume = {41},
  issue = {1},
  pages = {11--20},
  numpages = {0},
  year = {1990},
  month = {Jan},
  publisher = {American Physical Society},
  doi = {10.1103/PhysRevA.41.11},
  url = {https://link.aps.org/doi/10.1103/PhysRevA.41.11}
}

@article{degosson2012weak,
	doi = {10.1016/j.physleta.2011.11.007},
	url = {https://doi.org/10.1016/j.physleta.2011.11.007},
	year = 2012,
	month = {jan},
	publisher = {Elsevier {BV}},
	volume = {376},
	number = {4},
	pages = {293--296},
	author = {Maurice A. de Gosson and Serge M. de Gosson},
	title = {{Weak values of a quantum observable and the cross-Wigner distribution}},
	journal = {Phys. Lett. A}
}

@incollection{wigderson2019mathematics,
  title={Mathematics and computation},
  author={Wigderson, Avi},
  booktitle={Mathematics and Computation},
  year={2019},
  publisher={Princeton University Press},
  address={Princeton, NJ},
  url={https://press.princeton.edu/books/hardcover/9780691189130/mathematics-and-computation}
}

@article{higgins2015using,
	doi = {10.1103/physreva.91.012113},
	url = {https://doi.org/10.1103/physreva.91.012113},
	year = 2015,
	month = {jan},
	publisher = {American Physical Society ({APS})},
	volume = {91},
  pages = {012113},
  numpages = {8},
	number = {1},
	author = {B. L. Higgins and M. S. Palsson and G. Y. Xiang and H. M. Wiseman and G. J. Pryde},
	title = {Using weak values to experimentally determine {\textquotedblleft}negative probabilities{\textquotedblright} in a two-photon state with {B}ell correlations},
	journal = {Phys. Rev. A}
}

@article{lund2010measuring,
	doi = {10.1088/1367-2630/12/9/093011},
	url = {https://doi.org/10.1088/1367-2630/12/9/093011},
	year = 2010,
	month = {sep},
	publisher = {{IOP} Publishing},
	volume = {12},
	number = {9},
	pages = {093011},
	author = {A P Lund and H M Wiseman},
	title = {Measuring measurement{\textendash}disturbance relationships with weak values},
	journal = {New J. Phys.}
}

@article{bartlet2012reconstruction,
  title = {Reconstruction of Gaussian quantum mechanics from Liouville mechanics with an epistemic restriction},
  author = {Bartlett, Stephen D. and Rudolph, Terry and Spekkens, Robert W.},
  journal = {Phys. Rev. A},
  volume = {86},
  issue = {1},
  pages = {012103},
  numpages = {25},
  year = {2012},
  month = {Jul},
  publisher = {American Physical Society},
  doi = {10.1103/PhysRevA.86.012103},
  url = {https://link.aps.org/doi/10.1103/PhysRevA.86.012103}
}

@article{spekkens2007evidence,
  title = {Evidence for the epistemic view of quantum states: A toy theory},
  author = {Spekkens, Robert W.},
  journal = {Phys. Rev. A},
  volume = {75},
  issue = {3},
  pages = {032110},
  numpages = {30},
  year = {2007},
  month = {Mar},
  publisher = {American Physical Society},
  doi = {10.1103/PhysRevA.75.032110},
  url = {https://link.aps.org/doi/10.1103/PhysRevA.75.032110}
}

@article{dziewior2019universality,
	doi = {10.1073/pnas.1812970116},
	url = {https://doi.org/10.1073/pnas.1812970116},
	year = 2019,
	month = {feb},
	publisher = {Proc. Natl. Acad. Sci. U.S.A. },
	volume = {116},
	number = {8},
	pages = {2881--2890},
	author = {Jan Dziewior and Lukas Knips and Demitry Farfurnik and Katharina Senkalla and Nimrod Benshalom and Jonathan Efroni and Jasmin Meinecke and Shimshon Bar-Ad and Harald Weinfurter and Lev Vaidman},
	title = {Universality of local weak interactions and its application for interferometric alignment},
	journal = {Proc. Natl. Acad. Sci. U.S.A. }
}

@article{tamir2013introduction,
	doi = {10.12743/quanta.v2i1.14},
	url = {https://doi.org/10.12743/quanta.v2i1.14},
	year = 2013,
	month = {may},
	publisher = {Quanta},
	volume = {2},
	number = {1},
	pages = {7},
	author = {Boaz Tamir and Eliahu Cohen},
	title = {Introduction to Weak Measurements and Weak Values},
	journal = {Quanta}
}

@article{wagner2024quantumcircuits, title={Quantum circuits for measuring weak values, Kirkwood–Dirac quasiprobability distributions, and state spectra}, volume={9}, ISSN={2058-9565}, url={http://dx.doi.org/10.1088/2058-9565/ad124c}, DOI={10.1088/2058-9565/ad124c}, number={1}, journal={Quantum Science and Technology}, publisher={IOP Publishing}, author={Wagner, Rafael and Schwartzman-Nowik, Zohar and Paiva, Ismael L and Te’eni, Amit and Ruiz-Molero, Antonio and Barbosa, Rui Soares and Cohen, Eliahu and Galvão, Ernesto F}, year={2024}, month=Jan, pages={015030} }

@article{fernandes2024unitary,
  title = {{Unitary-Invariant Witnesses of Quantum Imaginarity}},
  author = {Fernandes, Carlos and Wagner, Rafael and Novo, Leonardo and Galv\~ao, Ernesto F.},
  journal = {Phys. Rev. Lett.},
  volume = {133},
  issue = {19},
  pages = {190201},
  numpages = {7},
  year = {2024},
  month = {Nov},
  publisher = {American Physical Society},
  doi = {10.1103/PhysRevLett.133.190201},
  url = {https://link.aps.org/doi/10.1103/PhysRevLett.133.190201}
}

@article{wagner2024coherence,
  doi = {10.22331/q-2024-02-05-1240},
  url = {https://doi.org/10.22331/q-2024-02-05-1240},
  title = {Coherence and contextuality in a {M}ach-{Z}ehnder interferometer},
  author = {Wagner, Rafael and Camillini, Anita and Galv{\~{a}}o, Ernesto F.},
  journal = {{Quantum}},
  issn = {2521-327X},
  publisher = {{Verein zur F{\"{o}}rderung des Open Access Publizierens in den Quantenwissenschaften}},
  volume = {8},
  pages = {1240},
  month = feb,
  year = {2024}
}

@article{wagner2024inequalities,
  title = {Inequalities witnessing coherence, nonlocality, and contextuality},
  author = {Wagner, Rafael and Barbosa, Rui Soares and Galv\~ao, Ernesto F.},
  journal = {Phys. Rev. A},
  volume = {109},
  issue = {3},
  pages = {032220},
  numpages = {18},
  year = {2024},
  month = {Mar},
  publisher = {American Physical Society},
  doi = {10.1103/PhysRevA.109.032220},
  url = {https://link.aps.org/doi/10.1103/PhysRevA.109.032220}
}

@article{giordani2023experimental, title={Experimental certification of contextuality, coherence, and dimension in a programmable universal photonic processor}, volume={9}, ISSN={2375-2548}, url={http://dx.doi.org/10.1126/sciadv.adj4249}, DOI={10.1126/sciadv.adj4249}, number={44}, journal={Sci. Adv.}, publisher={American Association for the Advancement of Science (AAAS)}, author={Giordani, Taira and Wagner, Rafael and Esposito, Chiara and Camillini, Anita and Hoch, Francesco and Carvacho, Gonzalo and Pentangelo, Ciro and Ceccarelli, Francesco and Piacentini, Simone and Crespi, Andrea and Spagnolo, Nicolò and Osellame, Roberto and Galvão, Ernesto F. and Sciarrino, Fabio}, year={2023}, month=Nov, pages={eadj4249} }

@article{maosheng2026multistate,
  title = {Multistate imaginarity and coherence in qubit systems},
  author = {Li, Mao-Sheng and Wagner, Rafael and Zhang, Lin},
  journal = {Phys. Rev. A},
  volume = {113},
  issue = {1},
  pages = {012428},
  numpages = {20},
  year = {2026},
  month = {Jan},
  publisher = {American Physical Society},
  doi = {10.1103/tpgw-v6ht},
  url = {https://link.aps.org/doi/10.1103/tpgw-v6ht}
}

@phdthesis{wagner2025coherenceandcontextuality,
    title = {Coherence and contextuality as quantum resources},
    author = {Rafael Wagner},
    year = {2025},
    school = {University of Minho},
    type = {{Ph.D. Thesis}},
    eprint = {2511.16785},
    archivePrefix = {arXiv},
    primaryClass = {quant-ph},
    url = {https://arxiv.org/abs/2511.16785},
    doi = {https://doi.org/10.48550/arXiv.2511.16785}
}

@misc{wagner2026bargmannscenarios,
      title={{Bargmann Scenarios}}, 
      author={Rafael Wagner},
      year={2026},
      eprint={2604.18833},
      archivePrefix={arXiv},
      primaryClass={quant-ph},
      url={https://arxiv.org/abs/2604.18833},
      doi={
https://doi.org/10.48550/arXiv.2604.18833}
}

@article{schmid2021characterization,
  title = {{Characterization of Noncontextuality in the Framework of Generalized Probabilistic Theories}},
  author = {Schmid, David and Selby, John H. and Wolfe, Elie and Kunjwal, Ravi and Spekkens, Robert W.},
  journal = {PRX Quantum},
  volume = {2},
  issue = {1},
  pages = {010331},
  numpages = {12},
  year = {2021},
  month = {Feb},
  publisher = {American Physical Society},
  doi = {10.1103/PRXQuantum.2.010331},
  url = {https://link.aps.org/doi/10.1103/PRXQuantum.2.010331}
}

@article{selby2023accessible,
  title = {Accessible fragments of generalized probabilistic theories, cone equivalence, and applications to witnessing nonclassicality},
  author = {Selby, John H. and Schmid, David and Wolfe, Elie and Sainz, Ana Bel\'en and Kunjwal, Ravi and Spekkens, Robert W.},
  journal = {Phys. Rev. A},
  volume = {107},
  issue = {6},
  pages = {062203},
  numpages = {21},
  year = {2023},
  month = {Jun},
  publisher = {American Physical Society},
  doi = {10.1103/PhysRevA.107.062203},
  url = {https://link.aps.org/doi/10.1103/PhysRevA.107.062203}
}

@article{plavala2023general, title={{General probabilistic theories: An introduction}}, volume={1033}, ISSN={0370-1573}, url={http://dx.doi.org/10.1016/j.physrep.2023.09.001}, DOI={10.1016/j.physrep.2023.09.001}, journal={Phys. Rep.}, publisher={Elsevier BV}, author={Plávala, Martin}, year={2023}, month=Sept, pages={1–64} }

@article{hermens2018constraints, title={Constraints on macroscopic realism without assuming non-invasive measurability}, volume={63}, ISSN={1355-2198}, url={http://dx.doi.org/10.1016/j.shpsb.2017.11.003}, DOI={10.1016/j.shpsb.2017.11.003}, journal={ Stud. Hist. Philos. Sci. Part B:
Stud. Hist. Philos. of Mod. Phys.}, publisher={Elsevier BV}, author={Hermens, R. and Maroney, O.J.E.}, year={2018}, month=Aug, pages={50–64} }

@article{schmid2025shadowssubsystemsof,
  doi = {10.22331/q-2025-10-13-1880},
  url = {https://doi.org/10.22331/q-2025-10-13-1880},
  title = {Shadows and subsystems of generalized probabilistic theories: when tomographic incompleteness is not a loophole for contextuality proofs},
  author = {Schmid, David and Selby, John H. and Rossi, Vinicius P. and Baldij{\~{a}}o, Roberto D. and Sainz, Ana Bel{\'{e}}n},
  journal = {{Quantum}},
  issn = {2521-327X},
  publisher = {{Verein zur F{\"{o}}rderung des Open Access Publizierens in den Quantenwissenschaften}},
  volume = {9},
  pages = {1880},
  month = oct,
  year = {2025}
}

@article{correa2021apparent,
  title = {Apparent quantum paradoxes as simple interference: Quantum violation of the pigeonhole principle and exchange of properties between quantum particles},
  author = {Corr\^ea, Raul and Saldanha, Pablo L.},
  journal = {Phys. Rev. A},
  volume = {104},
  issue = {1},
  pages = {012212},
  numpages = {6},
  year = {2021},
  month = {Jul},
  publisher = {American Physical Society},
  doi = {10.1103/PhysRevA.104.012212},
  url = {https://link.aps.org/doi/10.1103/PhysRevA.104.012212}
}

@article{rittaud2013pigeonhole, title={{The Pigeonhole Principle, Two Centuries Before Dirichlet}}, volume={36}, ISSN={1866-7414}, url={http://dx.doi.org/10.1007/s00283-013-9389-1}, DOI={10.1007/s00283-013-9389-1}, number={2}, journal={Math. Intell.}, publisher={Springer Science and Business Media LLC}, author={Rittaud, Benoît and Heeffer, Albrecht}, year={2013}, month=Aug, pages={27–29} }

@article{waegell2023negativemass, title={Quantum reality with negative-mass particles}, volume={120}, ISSN={1091-6490}, url={http://dx.doi.org/10.1073/pnas.2018437120}, DOI={10.1073/pnas.2018437120}, number={32}, journal={Proc. Natl. Acad. Sci. U.S.A. }, publisher={National Academy of Sciences}, author={Waegell, Mordecai and Cohen, Eliahu and Elitzur, Avshalom and Tollaksen, Jeff and Aharonov, Yakir}, year={2023}, month=July,pages={e2018437120} }

@article{duprey2018quantum, title={{The Quantum Cheshire Cat effect: Theoretical basis and observational implications}}, volume={391}, ISSN={0003-4916}, url={http://dx.doi.org/10.1016/j.aop.2018.01.011}, DOI={10.1016/j.aop.2018.01.011}, journal={Ann. Phys.}, publisher={Elsevier BV}, author={Duprey, Q. and Kanjilal, S. and Sinha, U. and Home, D. and Matzkin, A.}, year={2018}, month=Apr, pages={1–15} }

@article{ravon2007three, title={The three-box paradox revisited}, volume={40}, ISSN={1751-8121}, url={http://dx.doi.org/10.1088/1751-8113/40/11/021}, DOI={10.1088/1751-8113/40/11/021}, number={11}, journal={ J. Phys. A: Math. Theor.}, publisher={IOP Publishing}, author={Ravon, Tamar and Vaidman, Lev}, year={2007}, month=Feb, pages={2873–2882} }

@article{kirkpatrick2007reply, title={{Reply to ‘The three-box paradox revisited’ by T Ravon and L Vaidman}}, volume={40}, ISSN={1751-8121}, url={http://dx.doi.org/10.1088/1751-8113/40/11/N01}, DOI={10.1088/1751-8113/40/11/n01}, number={11}, journal={ J. Phys. A: Math. Theor.}, publisher={IOP Publishing}, author={Kirkpatrick, K A}, year={2007}, month=Feb, pages={2883–2890} }

@article{chitambar2019quantum,
   title={Quantum resource theories},
   volume={91},
   ISSN={1539-0756},
   url={http://dx.doi.org/10.1103/RevModPhys.91.025001},
   DOI={10.1103/revmodphys.91.025001},
   number={2},
   journal={Rev. Mod. Phys.},
   publisher={American Physical Society (APS)},
   author={Chitambar, Eric and Gour, Gilad},
   year={2019},
   month={Apr}
}

@article{shor1999polynomial,
  title={Polynomial-time algorithms for prime factorization and discrete logarithms on a quantum computer},
  author={Shor, Peter W.},
  journal={SIAM review},
  volume={41},
  number={2},
  pages={303--332},
  year={1999},
  publisher={SIAM},
  doi={https://doi.org/10.1137/S0036144598347011}
}

@article{naseri2022entanglement,
  title = {{Entanglement and coherence in the Bernstein-Vazirani algorithm}},
  author = {Naseri, Moein and Kondra, Tulja Varun and Goswami, Suchetana and Fellous-Asiani, Marco and Streltsov, Alexander},
  journal = {Phys. Rev. A},
  volume = {106},
  issue = {6},
  pages = {062429},
  numpages = {13},
  year = {2022},
  month = {Dec},
  publisher = {American Physical Society},
  doi = {10.1103/PhysRevA.106.062429},
  url = {https://link.aps.org/doi/10.1103/PhysRevA.106.062429}
}

@article{catani2023whyinterference,
  doi = {10.22331/q-2023-09-25-1119},
  url = {https://doi.org/10.22331/q-2023-09-25-1119},
  title = {Why interference phenomena do not capture the essence of quantum theory},
  author = {Catani, Lorenzo and Leifer, Matthew and Schmid, David and Spekkens, Robert W.},
  journal = {{Quantum}},
  issn = {2521-327X},
  publisher = {{Verein zur F{\"{o}}rderung des Open Access Publizierens in den Quantenwissenschaften}},
  volume = {7},
  pages = {1119},
  month = sep,
  year = {2023}
}

@article{zhang2026reassessing,
  title = {{Reassessing the Boundary between Classical and Nonclassical for Individual Quantum Processes}},
  author = {Zhang, Yujie and Schmid, David and Y\ifmmode \bar{\imath}\else \={\i}\fi{}ng, Y\`{\i}l\`e and Spekkens, Robert W.},
  journal = {Phys. Rev. X},
  volume = {16},
  issue = {2},
  pages = {021050},
  numpages = {54},
  year = {2026},
  month = {Jun},
  publisher = {American Physical Society},
  doi = {10.1103/vqfz-wzjg},
  url = {https://link.aps.org/doi/10.1103/vqfz-wzjg}
}

@article{hall2004prior,
  title = {Prior information: How to circumvent the standard joint-measurement uncertainty relation},
  author = {Hall, Michael J. W.},
  journal = {Phys. Rev. A},
  volume = {69},
  issue = {5},
  pages = {052113},
  numpages = {12},
  year = {2004},
  month = {May},
  publisher = {American Physical Society},
  doi = {10.1103/PhysRevA.69.052113},
  url = {https://link.aps.org/doi/10.1103/PhysRevA.69.052113}
}

@article{hardy1999disentangling,
  title={Disentangling nonlocality and teleportation},
  author={Hardy, Lucien},
  journal={arXiv:quant-ph/9906123},
  year={1999},
  url={https://doi.org/10.48550/arXiv.quant-ph/9906123}
}

@article{vaidman1996weak, title={Weak-measurement elements of reality}, volume={26}, ISSN={1572-9516}, url={http://dx.doi.org/10.1007/BF02148832}, DOI={10.1007/bf02148832}, number={7}, journal={Foundations of Physics}, publisher={Springer Science and Business Media LLC}, author={Vaidman, Lev}, year={1996}, month=July, pages={895–906} }

@article{hance2024counterfactuality, title={Counterfactuality, back-action, and information gain in multi-path interferometers}, volume={9}, ISSN={2058-9565}, url={http://dx.doi.org/10.1088/2058-9565/ad63c7}, DOI={10.1088/2058-9565/ad63c7}, number={4}, journal={Quantum Science and Technology}, publisher={IOP Publishing}, author={Hance, Jonte R and Matsushita, Tomonori and Hofmann, Holger F}, year={2024}, month=July, pages={045015} }

@article{hofmann2020contextuality,
  title = {Contextuality of quantum fluctuations characterized by conditional weak values of entangled states},
  author = {Hofmann, Holger F.},
  journal = {Phys. Rev. A},
  volume = {102},
  issue = {6},
  pages = {062215},
  numpages = {7},
  year = {2020},
  month = {Dec},
  publisher = {American Physical Society},
  doi = {10.1103/PhysRevA.102.062215},
  url = {https://link.aps.org/doi/10.1103/PhysRevA.102.062215}
}

@article{ming2023characterization,
  title = {Characterization of the nonclassical relation between measurement outcomes represented by nonorthogonal quantum states},
  author = {Ji, Ming and Hofmann, Holger F.},
  journal = {Phys. Rev. A},
  volume = {107},
  issue = {2},
  pages = {022208},
  numpages = {8},
  year = {2023},
  month = {Feb},
  publisher = {American Physical Society},
  doi = {10.1103/PhysRevA.107.022208},
  url = {https://link.aps.org/doi/10.1103/PhysRevA.107.022208}
}

@article{catani2023aspects,
  title = {Aspects of the phenomenology of interference that are genuinely nonclassical},
  author = {Catani, Lorenzo and Leifer, Matthew and Scala, Giovanni and Schmid, David and Spekkens, Robert W.},
  journal = {Phys. Rev. A},
  volume = {108},
  issue = {2},
  pages = {022207},
  numpages = {11},
  year = {2023},
  month = {Aug},
  publisher = {American Physical Society},
  doi = {10.1103/PhysRevA.108.022207},
  url = {https://link.aps.org/doi/10.1103/PhysRevA.108.022207}
}

@article{ahnefeld2026coherence,
  title = {{Coherence as a Resource for Phase Estimation}},
  author = {Ahnefeld, Felix and Theurer, Thomas and Plenio, Martin B.},
  journal = {Phys. Rev. Lett.},
  volume = {136},
  issue = {18},
  pages = {180201},
  numpages = {8},
  year = {2026},
  month = {May},
  publisher = {American Physical Society},
  doi = {10.1103/s6d5-lc6g},
  url = {https://link.aps.org/doi/10.1103/s6d5-lc6g}
}

@article{ahnefeld2022role,
	doi = {10.1103/physrevlett.129.120501},
	url = {https://doi.org/10.1103/physrevlett.129.120501},
	year = 2022,
	month = {September},
	publisher = {American Physical Society ({APS})},
	volume = {129},
	number = {12},
	author = {Felix Ahnefeld and Thomas Theurer and Dario Egloff and Juan Mauricio Matera and Martin B. Plenio},
	title = {Coherence as a Resource for {S}hor's Algorithm},
	journal = {Phys. Rev. Lett.},
    pages={120501}
}

@article{wu2021experimental, title={{Experimental Progress on Quantum Coherence: Detection, Quantification, and Manipulation}}, volume={4}, ISSN={2511-9044}, url={http://dx.doi.org/10.1002/qute.202100040}, DOI={10.1002/qute.202100040}, number={9}, journal={Adv. Quantum Technol.}, publisher={Wiley}, author={Wu, Kang‐Da and Streltsov, Alexander and Regula, Bartosz and Xiang, Guo‐Yong and Li, Chuan‐Feng and Guo, Guang‐Can}, year={2021}, month=July,pages={2100040} }

@misc{aberg2006quantifyingsuperposition,
      title={{Quantifying Superposition}}, 
      author={Johan Aberg},
      year={2006},
      eprint={quant-ph/0612146},
      archivePrefix={arXiv},
      primaryClass={quant-ph},
      url={https://arxiv.org/abs/quant-ph/0612146}, 
}

@article{kirkpatrick2003classical, title={Classical three-box  paradox }, volume={36}, ISSN={1361-6447}, url={http://dx.doi.org/10.1088/0305-4470/36/17/315}, DOI={10.1088/0305-4470/36/17/315}, number={17}, journal={J. Phys. A: Math. Gen. }, publisher={IOP Publishing}, author={Kirkpatrick, K A}, year={2003}, month=Apr, pages={4891–4900} }

@article{harrigan2010Einstein, title={{Einstein, Incompleteness, and the Epistemic View of Quantum States}}, volume={40}, ISSN={1572-9516}, url={http://dx.doi.org/10.1007/s10701-009-9347-0}, DOI={10.1007/s10701-009-9347-0}, number={2}, journal={Found. Phys.}, publisher={Springer Science and Business Media LLC}, author={Harrigan, Nicholas and Spekkens, Robert W.}, year={2010}, month=Jan, pages={125–157} }

@article{pratapsi2025elementarycharacterizationbargmanninvariants,
  title = {Elementary characterization of Bargmann invariants},
  author = {Pratapsi, Sagar Silva and Gouveia, Jo\~ao and Novo, Leonardo and Galv\~ao, Ernesto F.},
  journal = {Phys. Rev. A},
  volume = {112},
  issue = {4},
  pages = {042421},
  numpages = {7},
  year = {2025},
  month = {Oct},
  publisher = {American Physical Society},
  doi = {10.1103/hsnv-wpt3},
  url = {https://link.aps.org/doi/10.1103/hsnv-wpt3}
}

@article{zhang2025geometrysets,
  title = {{Geometry of sets of Bargmann invariants}},
  author = {Zhang, Lin and Xie, Bing and Li, Bo},
  journal = {Phys. Rev. A},
  volume = {111},
  issue = {4},
  pages = {042417},
  numpages = {12},
  year = {2025},
  month = {Apr},
  publisher = {American Physical Society},
  doi = {10.1103/PhysRevA.111.042417},
  url = {https://link.aps.org/doi/10.1103/PhysRevA.111.042417}
}

@article{zhang2024local,
  title = {Bargmann-invariant framework for local unitary equivalence and entanglement},
  author = {Zhang, Lin and Xie, Bing and Tao, Yuanhong},
  journal = {Phys. Rev. A},
  volume = {112},
  issue = {5},
  pages = {052426},
  numpages = {28},
  year = {2025},
  month = {Nov},
  publisher = {American Physical Society},
  doi = {10.1103/s3mp-3kn6},
  url = {https://link.aps.org/doi/10.1103/s3mp-3kn6}
}

@article{zamora2025semi,
  title = {Semi-device-independent nonstabilizerness certification in the prepare-and-measure scenario},
  author = {Zamora, Santiago and Mac\^edo, Rafael A. and Sarubi, Tailan S. and Alves, Mois\'es and Poderini, Davide and Chaves, Rafael},
  journal = {Phys. Rev. A},
  volume = {112},
  issue = {4},
  pages = {042410},
  numpages = {14},
  year = {2025},
  month = {Oct},
  publisher = {American Physical Society},
  doi = {10.1103/wm7m-xnfq},
  url = {https://link.aps.org/doi/10.1103/wm7m-xnfq}
}

@article{jones2023distinguishability,
  title = {Distinguishability and mixedness in quantum interference},
  author = {Jones, Alex E. and Kumar, Shreya and D'Aurelio, Simone and Bayerbach, Matthias and Menssen, Adrian J. and Barz, Stefanie},
  journal = {Phys. Rev. A},
  volume = {108},
  issue = {5},
  pages = {053701},
  numpages = {11},
  year = {2023},
  month = {Nov},
  publisher = {American Physical Society},
  doi = {10.1103/PhysRevA.108.053701},
  url = {https://link.aps.org/doi/10.1103/PhysRevA.108.053701}
}

@article{grassl1998computing,
  title = {Computing local invariants of quantum-bit systems},
  author = {Grassl, Markus and R\"otteler, Martin and Beth, Thomas},
  journal = {Phys. Rev. A},
  volume = {58},
  issue = {3},
  pages = {1833--1839},
  numpages = {0},
  year = {1998},
  month = {Sep},
  publisher = {American Physical Society},
  doi = {10.1103/PhysRevA.58.1833},
  url = {https://link.aps.org/doi/10.1103/PhysRevA.58.1833}
}

@article{wagner2024certifying,
  author    = {Wagner, Rafael
               and Peres, Filipa C.\ R.
               and Zambrini Cruzeiro, Emmanuel
               and Galv\~{a}o, Ernesto F.},
  title     = {{Unitary‐invariant method for witnessing nonstabilizerness in quantum processors}},
  journal   = { J. Phys. A: Math. Theor.},
  volume    = {58},
  number    = {28},
  pages     = {285302},
  year      = {2025},
  month     = jul,
  doi       = {10.1088/1751-8121/ade9ff},
  issn      = {1751-8121},
  publisher = {IOP Publishing},
  url       = {https://doi.org/10.1088/1751-8121/ade9ff},
}

@article{xu2026bargmanninvariantsgaussianstates, title={{Bargmann invariants of Gaussian states}}, volume={67}, ISSN={1089-7658}, url={http://dx.doi.org/10.1063/5.0306797}, DOI={10.1063/5.0306797}, number={5}, journal={J. Math. Phys.}, publisher={AIP Publishing}, author={Xu, Jianwei}, year={2026}, month=May,pages={052101} }

@misc{zhang2026surveybargmanninvariantsgeometric,
      title={{A Survey of Bargmann Invariants: Geometric Foundations and Applications}}, 
      author={Lin Zhang and Bing Xie},
      year={2026},
      eprint={2601.01858},
      archivePrefix={arXiv},
      primaryClass={quant-ph},
      url={https://arxiv.org/abs/2601.01858}, 
}

\end{document}